\documentclass{patmorin}
\usepackage{amsthm,amsmath,graphicx,wrapfig}
\usepackage{pat}
\usepackage[letterpaper]{hyperref}
\usepackage[dvipsnames]{color}
\definecolor{linkblue}{named}{Blue}
\hypersetup{colorlinks=true, linkcolor=linkblue,  anchorcolor=linkblue,
citecolor=linkblue, filecolor=linkblue, menucolor=linkblue, pagecolor=linkblue,
urlcolor=linkblue}

\usepackage{tikz,hyphenat}
\let\oldmarginpar\marginpar
\renewcommand{\marginpar}[2][rectangle,draw,fill=yellow,rounded corners,text width=2.21cm]{%
        \oldmarginpar{%
        \tikz \node at (0,0) [#1]{#2};}%
        }

\DeclareMathOperator{\asf}{asf}
\DeclareMathOperator{\xsf}{(a)sf}
\DeclareMathOperator{\strf}{sf}
\DeclareMathOperator{\depth}{depth}
\DeclareMathOperator{\radius}{radius}
\DeclareMathOperator{\msst}{msst}

\newcommand{\eps}{\varepsilon}

\title{\MakeUppercase{Average Stretch Factor: How Low Does It Go?}}
\author{Vida Dujmovi\'c, Pat Morin, and Michiel Smid}

\begin{document}
\begin{titlepage}
\maketitle

\begin{abstract}
  In a geometric graph, $G$, the \emph{stretch factor} between two
  vertices, $u$ and $w$, is the ratio between the Euclidean length of
  the shortest path from $u$ to $w$ in $G$ and the Euclidean distance
  between $u$ and $w$.  The \emph{average stretch factor} of $G$ is
  the average stretch factor taken over all pairs of vertices in $G$.
  We show that, for any constant dimension, $d$, and any set, $V$, of
  $n$ points in $\R^d$, there exists a geometric graph with vertex set
  $V$, that has $O(n)$ edges, and that has average stretch factor $1+
  o_n(1)$.  More precisely, the average stretch factor of this graph is
  $1+O((\log n/n)^{1/(2d+1)})$.  We complement this upper-bound with a
  lower bound: There exist $n$-point sets in $\R^2$ for which any graph
  with $O(n)$ edges has average stretch factor $1+\Omega(1/\sqrt{n})$.
  Bounds of this type are not possible for the more commonly studied
  worst-case stretch factor.  In particular, there exists point sets,
  $V$, such that any graph with worst-case stretch factor $1+o_n(1)$
  has a superlinear number of edges.
\end{abstract}

\end{titlepage}

\section{Introduction}

A \emph{geometric graph} is a simple undirected graph whose vertex set
is a set of points in $\R^d$.  The \emph{average stretch factor} of
a finite connected geometric graph, $G=(V,E)$, with vertex set $V\subset \R^d$
and edge set $E$ is
\begin{equation}
    \asf(G) = \binom{n}{2}^{-1}\sum_{\{u,w\}\in\binom{V}{2}}\frac{\|uw\|_G}{\|uw\|} \eqlabel{asf}
\end{equation}
where $\|uw\|$ denotes the Euclidean distance between $u$ and $w$
and $\|uw\|_G$ denotes the cost of the shortest path from $u$ to $w$
in the graph $G$, where each edge of $G$ is weighted by the Euclidean
distance between its endpoints.  (Here, and forever, $n=|V|$.)

The related notion of \emph{worst-case
stretch factor}, defined as
\[
    \strf(G) = \max_{\{u,w\}\in\binom{V}{2}}\frac{\|uw\|_G}{\|uw\|}  \enspace ,
\]
has been studied extensively.  A graph, $G$, with $\strf(G)\le t$ is
called a \emph{$t$-spanner}.  The construction of and applications
of $t$-spanners is the subject of intensive research and there
is a book \cite{narasimhan.smid:geometric} and handbook chapter
\cite{eppstein:spanning} devoted to the topic.

The average or worst-case stretch factor of the complete graph
with vertex set $V\subset\R^d$ is 1 and, if $V$ contains no
collinear triples, then any graph on $V$ with fewer than $\binom{n}{2}$
edges will have (average and worst-case) stretch factor strictly greater
than 1.  At the other extreme, any connected graph with vertex set $V$
has at least $n-1$ edges.  Thus, there seems to be a tradeoff between
the following two requirements on a graph $G=(V,E)$ that is constructed
from a given point set, $V$:
\begin{enumerate}
  \item $G$ should be as sparse as possible, ideally $|E|\in O(n)$; and

  \item $G$ should have small (average or worst-case) stretch factor,
  ideally $\xsf(G)=1+o_n(1)$.
\end{enumerate}
In this paper we are interested in determining to what extent we can
simultaneously satisfy these two conflicting goals.

For worst-case stretch factor, it is not possible to achieve the preceding
two goals simultaneously: Elkin and Solomon \cite{elkin.solomon:steiner}
show that there exists point sets $V\subset\R^d$ such that any
graph $G=(V,E)$ having $\strf(G)\le 1+\eps$ has maximum degree
$\Omega(1/\eps^{d-1})$.  A variant of their argument shows that,
for some point sets, a constant fraction of vertices must have
degree $\Omega(1/\eps^{d-1})$ and therefore the graph must have
$\Omega(n/\eps^{d-1})$ edges \cite{solomon}. In short: achieving a
worst-case stretch factor of $1+o_n(1)$ requires $\omega(n)$ edges.



\subsection{New Results and Outline}

In this paper we show that, for average stretch factor, it \emph{is}
possible to simultaneously achieve both desired properties: For any
constant dimension, $d$, and any set $V\subset \R^d$, there exists
a geometric graph $G=(V,E)$ having $|E|\in O(n)$ edges and such that
$\asf(G)=1+o_n(1)$.  More precisely,
\[
   \asf(G)=1+O((\log n/n)^{1/(2d+1)}) \enspace .
\]
The proof of this result is in \secref{upper-bound} and constitutes the
bulk of the paper.  

In \secref{algorithms}, we show that graphs with small average stretch
factor can be constructed efficiently.  In particular, we present an
$O(n\log n)$ time Monte-Carlo algorithm that, with high probability,
constructs a graph with average stretch factor $1+o_n(1)$.  
In \secref{lower-bound} we present a simple lower-bound that shows our
upper-bound is at least of the right flavour:  For every positive integer
$n$, there exists an $n$ point set in $\R^2$ on which any graph with
$O(n)$ edges has average stretch factor $1+\Omega(1/\sqrt{n})$.
In \secref{discussion} we relate the parts of our construction to some
real-world networks and discuss directions for future work.

\subsection{Relation to Previous Work}


In the more general context of embedding metric spaces into ultrametrics,
Abraham \etal\ \cite{abraham.bartal.ea:metric,abraham.bartal.ea:embedding}
show that for any point set, $V$, there exists a spanning tree,
$T=(V,E)$, with $\asf(T)\in O(1)$.  Thus, it is always possible to
construct a very sparse graph with constant average stretch factor.
This contrasts sharply with worst-case stretch factor: Any spanning
tree on the vertices of a regular $n$-gon has worst-case stretch factor
$\Omega(n)$ \cite[Lemma~15]{eppstein:spanning-report}.

Aldous and Kendall \cite[Section~5.3]{aldous.kendall:short-length}
show that, for any well-distributed\footnote{For example, any
family of point sets that satisfies the quantitative equidistribution
condition \cite[Definition~3]{aldous.kendall:short-length}.} point set
$V\subset[0,\sqrt{n}]^2$ and any $\eps > 0$ there exists a Steiner
network, $N=(V',E)$, with $V'\supseteq V$, of total edge length
$\sum_{uw\in E}\|uw\| \le (1+\eps)\msst(V)$ and for which,
\[
    \binom{n}{2}^{-1}\sum_{\{u,w\}\in \binom{V}{2}}\frac{\|uw\|_N}{\|uw\|} = 1 + O((\log n/n)^{1/3}) .
\]
Here $\msst(V)$ denotes the length of the minimum Steiner spanning tree
of the points in $V$.  Thus $N$ is a graph that is only slightly longer
than the minimum length connected graph that contains $V$ and the average
stretch factor of $N$ (taken over pairs in $V$) tends to 1.  In this
construction, we call the points in $V'\setminus V$ \emph{Steiner points}.

In the work of Aldous and Kendall, which was the starting point for
the current work, the authors focus on finding a \emph{light} network;
one whose total edge length is small.  In the current paper, we focus
instead on finding a \emph{sparse} network; one with a small number
of edges.  With this shift of focus in mind, Aldous and Kendall's
work immediately raises three questions: (1)~Can a similar result be
proven without making any form of ``well-distributed'' assumption on the
points in $V$? (2)~Can a similar result be proven without using Steiner
points? (3)~Can a similar result be proven for point sets in $\R^d$?
Our results answer all three of these questions in the affirmative.

We note that it does not seem easy to answer any of the preceding
questions using a modification of Aldous and Kendall's construction.
(1)~Their proof uses the well-distributed assumption to argue that
most pairs of points are at distance at least $n^{\gamma}$, for any
$\gamma<1/2$ \cite[Section~5.3]{aldous.kendall:short-length}.  (2)~their
construction consists of a minimum Steiner spanning tree of $V$,
some additional random line segments, and some additional segments
that form a grid.  Anywhere two segments cross, a vertex is added
to $V'$, so this construction makes essential use of Steiner points.
(3)~The main technical tool used in their proof is a new result on the
lengths of boundaries of certain cells in arrangements of random lines
\cite[Theorems~3 and 4]{aldous.kendall:short-length}.  In dimensions
greater than 2, arrangements of lines do not decompose space into cells,
so it seems difficult to generalize this result to higher dimensions.

\section{The Construction}
\seclabel{upper-bound}

Our construction of a good average stretch factor graph, $G=(V,E)$,
makes use of a clustering of the points of $V$ into $O(n/k)$ clusters,
each of size at most $k$, that we call a $k$-partition.  In the next
subsection, we define $k$-partitions and show how to compute them.
In the subsequent subsection we show how to construct the graph $G$.


\subsection{$k$-Partitions}

We make use of the following construct:  A \emph{$k$-partition} of a
set $V$ of $n$ points in $\R^d$ consists of a set, $D$, of balls and an
assignment $f:V\to D$ such that
\begin{enumerate}
  \item $|D|\in O(n/k)$;
  \item for each $u\in V$, $u\in f(u)$ (i.e., $u$ is assigned to a
    ball that contains $u$);
  \item for each $\Delta\in D$, $|\{u\in V: f(u)=\Delta\}|\le k$ (i.e.,
   at most $k$ points are assigned to each ball);
  \item for every $r> 0$ and $p\in\R^d$, the number of balls
   in $D$ whose radius is in the range $[r,2r)$ and that contain $p$
   is $O(1)$; and
  \item for every $r\ge 0$ and every ball, $B$, of radius $r$, 
   \[
      |\{ u\in V : u\in B\text{ and } \radius(f(u))\ge r\}| \in O(k)
   \] 
   (i.e., there are only $O(k)$ points
   of $V$ that are in $B$ and that are assigned to balls of radius at
   least $r$).
\end{enumerate}

Note that, aside from Properties~4 and 5, there is very little structure
to the balls in $D$. In particular, balls in $D$ may overlap and may
even contain each other.  

\begin{lem}\lemlabel{k-partition}
  For any constant dimension, $d$, and any set $V$ of $n$ points in
  $\R^d$, a $k$-partition of $V$ exists and can be computed in $O(n\log
  n)$ time.
\end{lem}

\begin{proof}
  We construct a $k$-partition using the binary \emph{fair-split
  tree}, $T=T(V)$, which is defined recursively as follows
  \cite{callahan.kosaraju:decomposition}: If $V$ consists of a single
  point, $u$, then $T$ contains a single node corresponding to $u$.
  Otherwise, consider the minimal axis-aligned bounding box, $B(V)$,
  that contains $V$.  The root of $T$ corresponds to $B(V)$ and this
  box is split into two boxes $B_1(V)$ and $B_2(V)$ by cutting $B(V)$
  with a hyperplane in the middle of its longest side.  The left and
  right subtrees of the root are defined recursively by constructing
  fair-split trees for $B_1(V)\cap V$ and $B_2(V)\cap V$. See \figref{fst}.

  \begin{figure}
    \begin{center}
      \includegraphics{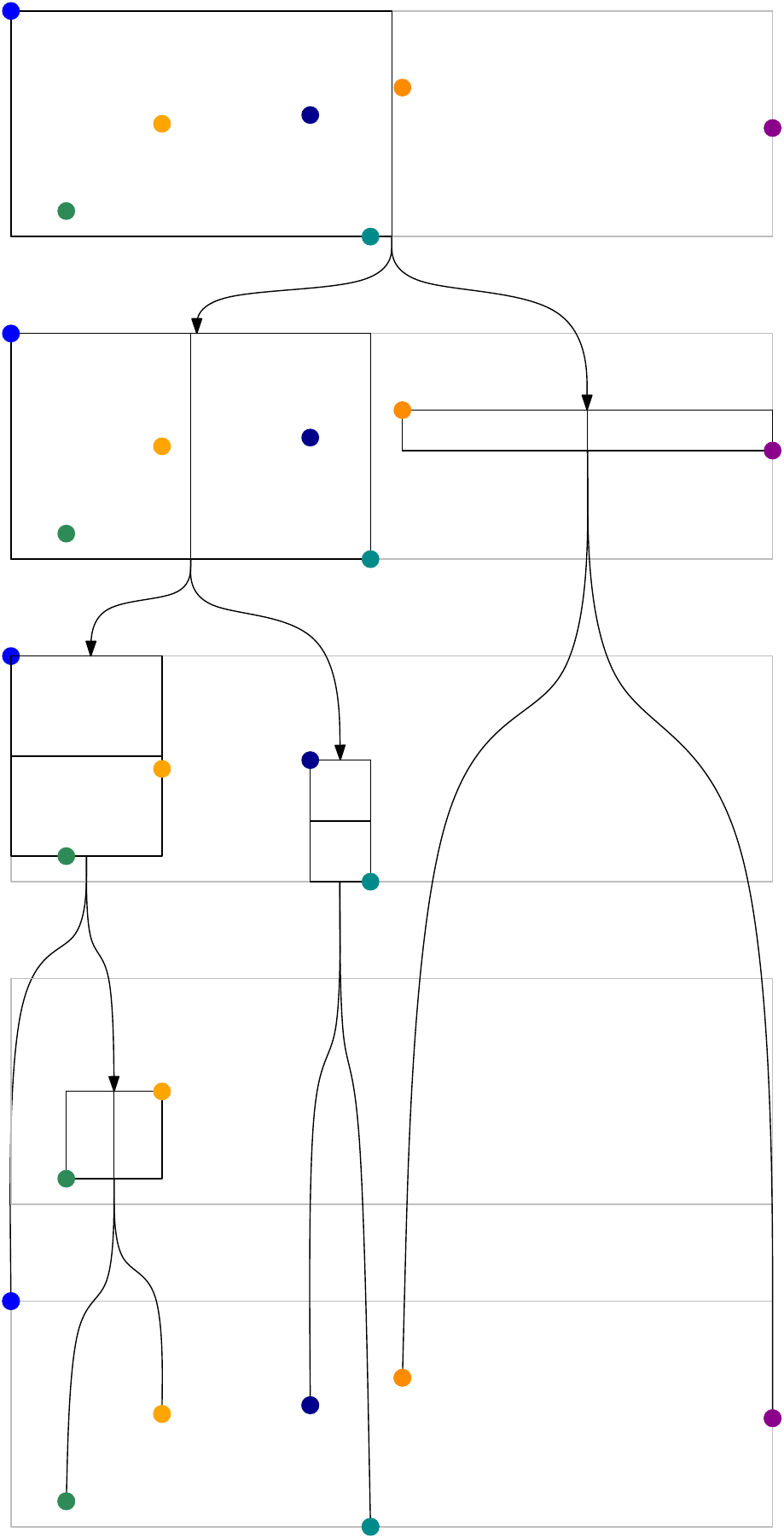}
    \end{center}
    \caption{A fair-split tree for $V$ repeatedly splits the bounding
      box $B(V)$ in the middle of its longest side.}
    \figlabel{fst}
  \end{figure}

  \paragraph{The $k$-Partition.}
  For each node, $u$, of $T$ there is a naturally defined subset
  $V(u)\subseteq V$ of points associated with $u$ as well as a bounding
  box $B(u)=B(V(u))$.  Since $T$ is a binary tree with $2n-1$ nodes, it
  has a set of $t-1$ edges whose removal partitions the vertices of $T$
  into $t\in O(n/k)$ maximally-connected components $C_1,\ldots,C_t$,
  each having at most $k$ vertices.

  For each $i\in\{1,\ldots,t\}$, let $u_i$ denote the node in $C_i$ of
  minimum depth, and call $u_i$, the \emph{root} of $C_i$.
  To obtain the balls, $\Delta_1,\ldots,\Delta_t$, of the $k$-partition we
  take, for each $i\in\{1,\ldots,t\}$ the smallest ball, $\Delta_i$ that
  contains $B(u_i)$.  For the mapping $f$, we map the point associated
  with each leaf, $w$, of $T$ to the unique ball $\Delta_i$, where $C_i$
  contains $w$. See \figref{fst-2}.

  \begin{figure}
    \begin{center}
      \includegraphics{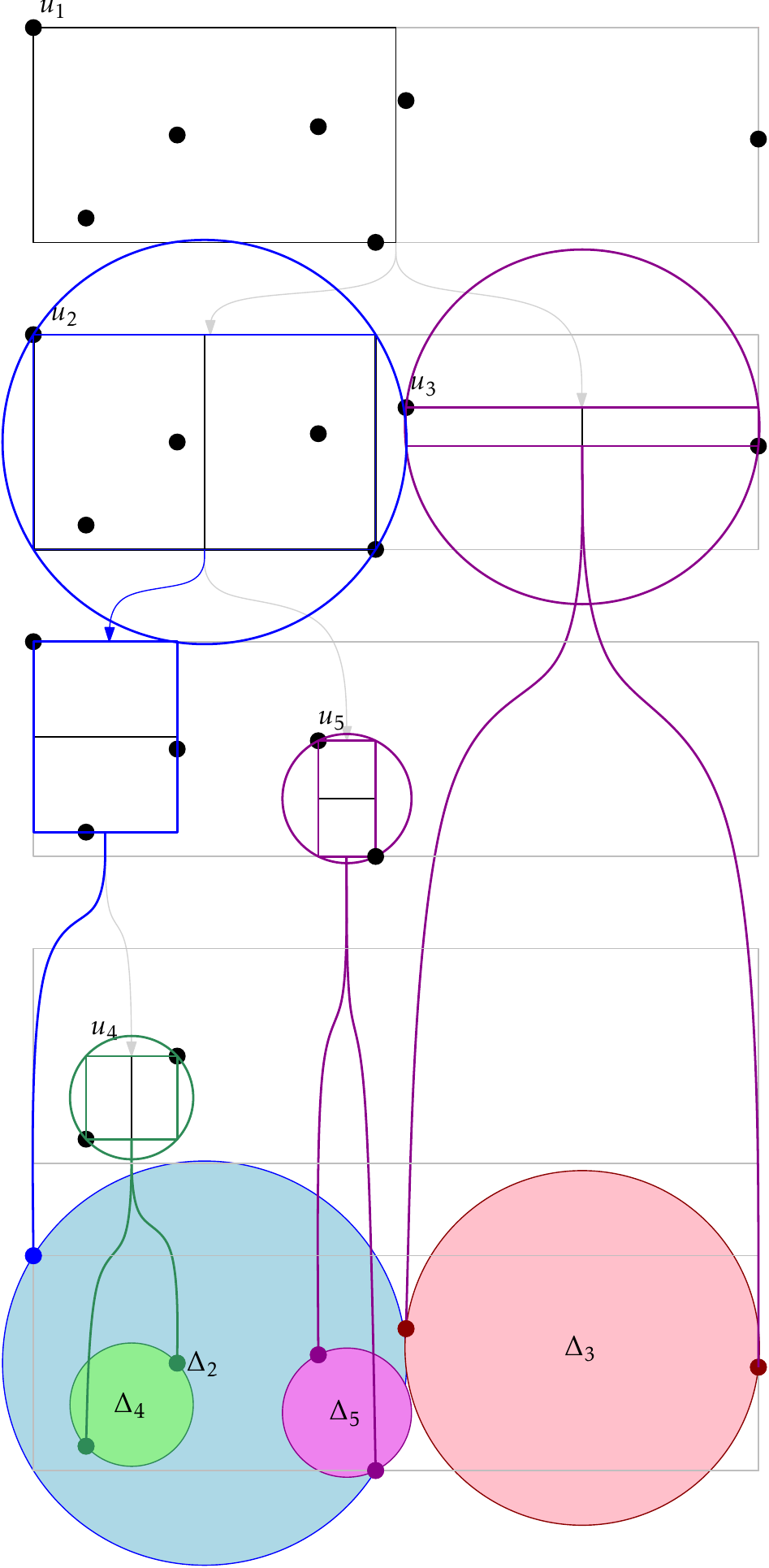}
    \end{center}
    \caption{The fair-split tree is partitioned into subtrees of size
    $k$ (${}=3$) by removing $O(n/k)$ edges.  The root, $u_i$, of each subtree
    defines a ball, $\Delta_i$, in the $k$-partition. (The ball $\Delta_1$
    is omitted from this figure.)}
    \figlabel{fst-2}
  \end{figure}

  The fair-split tree, $T$, and the boxes, $B(u)$, associated
  with each node, $u$, of $T$ can be computed in $O(n\log n)$ time
  \cite{callahan.kosaraju:decomposition}.  The partition of the vertices
  of $T$ into components $C_1,\ldots,C_t$ can easily be done in $O(n\log
  n)$ time by repeatedly finding an edge of a component of size $k'>k$
  whose removal partitions that component into two pieces each of size
  at most $\lceil 2k'/3\rceil$.  Thus, the construction of $D$ and $f$
  can be accomplished in $O(n\log n)$ time.

  The set of balls $D=\{\Delta_1,\ldots,\Delta_t\}$ and the mapping
  $f:V\to D$ described in the preceding paragraphs clearly satisfy
  Properties~1--3 in the definition of a $k$-partition.  What remains
  is to show that they also satisfy Properties~4 and 5. 

  For a node $u$ in $T$, with $B(u)=[a_1,b_1]\times\cdots\times[a_d,b_d]$,
  define $L_i(u)=b_i-a_i$ and let $L(u)=\max\{L_i(i):i\in\{1,\ldots,d\}\}$
  denote the length of $B(u)$'s longest side.   To establish
  Properties~4 and 5, we make use of the following result on fair-split
  trees \cite[Lemma~9.4.3]{narasimhan.smid:geometric}:

  \begin{lem}\lemlabel{box-packing}
     Let $C$ be a box whose longest side has length $\ell$ and let
     $\alpha >0$ be any positive real number.  Let $w_1,\ldots,w_s$
     be some nodes of a fair-split tree, $T$, such that
     \begin{enumerate}
       \item the sets $V(w_i)$ are pairwise disjoint, for all $i\in\{1,\ldots,s\}$;
       \item $L(w_i)\ge \ell/\alpha$, for all
          $i\in\{1,\ldots,s\}$;\footnote{The original lemma
          \cite[Lemma~9.4.3]{narasimhan.smid:geometric} is slightly
          stronger in that it only requires that $L(w_i')\ge \ell/\alpha$,
          where $w_i'$ is the parent of $w_i$.} and
       \item $B(w_i)$ intersects $C$, for all $i\in\{1,\ldots,s\}$.
     \end{enumerate}
     Then $s\le (2\alpha + 2)^d$.
  \end{lem}

  \paragraph{Property 4.}
  To prove that the balls in $D$ satisfy Property~4, let
  $\{\Delta_{i_1},\ldots,\Delta_{i_q}\}\subseteq D$ be the subset of
  balls in $D$ having radii in the interval $[r,2r)$ and that all contain
  some common point, $p\in\R^d$.   Then each such ball, $\Delta_{i_j}$
  corresponds to a node $u_{i_j}$ of $T$ such that
  \begin{equation}
        \frac{2r}{\sqrt{d}} \le L(u_{i_j}) \le 4r \enspace . \eqlabel{bounds}
  \end{equation}
  Each box, $B(u_{i_j})$, intersects the ball of radius $2r$ centered at
  $p$.  (Indeed, the center of $B(u_{i_j})$ is contained in this ball.)
  Therefore, each box $B(u_{i_j})$ intersects the box, $C$, of side-length
  $4r$ centered at $p$.

  We are almost ready to apply \lemref{box-packing} to
  $C$---whose side length is $\ell = 4r$---and the vertex set
  $w_1,\ldots,w_q=u_{i_1},\ldots,u_{i_q}$. For each $j\in\{1,\ldots,q\}$,
  $B(u_{i_j})$ intersects $C$, so Condition~3 of \lemref{box-packing}
  is satisfied.  Furthermore, \eqref{bounds} states that $L(u_{i_j})\ge
  2r/\sqrt{d} = \ell/(2\sqrt{d})$, so Condition~2 of \lemref{box-packing}
  is satisfied with $\alpha=2\sqrt{d}$.  Unfortunately, there is still
  a little more work to do since the nodes $u_{i_1},\ldots,u_{i_t}$
  do not necessarily satisfy Condition~1 of \lemref{box-packing}.

  To proceed, we partition $u_{i_1},\ldots,u_{i_q}$ into a small number of
  subsets, each of which satisfies Condition~1 of \lemref{box-packing}.
  Observe that Condition~1 of \lemref{box-packing} is equivalent
  to the statement that no $w_i$ is an ancestor of $w_j$ for any
  $\{i,j\}\subseteq\{1,\ldots,s\}$.  A key observation is that, if $u$
  is an ancestor of $w$ in a fair-split tree, $T$, and the difference
  in depth between $u$ and $w$ is at least $d$, then
  \[
      L(u) \ge 2L(w) \enspace .
  \]
  This, and \eqref{bounds}, implies that, if $u_{i_j}$ is an ancestor
  of $u_{i_{j'}}$ then
  \[
     \depth(u_{i_{j'}})-\depth(u_{i_{j}}) \le d\log(2\sqrt{d}) \enspace .
  \]
  Thus, we can partition $u_{i_1},\ldots,u_{i_q}$ into $z=\lceil
  d\log(2\sqrt{d})\rceil$ subsets, $S_0,\ldots,S_{z-1}$, each of which
  satisfies Condition~1 of \lemref{box-packing}, by assigning $u_{i_j}$
  to the subset $S_{\depth(u_{i_j})\bmod z}$.  

  Now, \lemref{box-packing} implies that, for each $i\in\{0,\ldots,z-1\}$, 
  \[
     |S_i|\le (4\sqrt{d}+2)^d
  \]
  so that
  \[
     q = \sum_{i=0}^{z-1}|S_i|\le (4\sqrt{d}+2)^dz 
       = (4\sqrt{d}+2)^d\lceil d\log(2\sqrt{d})\rceil  
       \in O(1) \enspace .
  \]
  Thus, for any point $p\in\R^d$, the set of balls in $D$ whose radius
  is in the interval $[r,2r)$ and that contain $p$ has size $O(1)$.
  Therefore the balls in $D=\{\Delta_1,\ldots,\Delta_t\}$ satisfy
  Property~4 in the definition of a $k$-partition.

  \paragraph{Property 5.}
  To study Property~5, it is easier to work with the bounding boxes,
  $B(u)$, associated with each node, $u$, in the fair-split tree as
  well as the box, $C$, of side length $2r$ that contains the ball $B$.
  See \figref{property-5}.  Observe that if some ball, $\Delta_i$,
  is assigned a point in $B$, then the box $B(u_i)$ intersects $C$.
  Thus, we need only consider the set $U\subseteq\{u_1,\ldots,u_t\}$
  that contains only those nodes $u_i$ such that $\radius(\Delta_i)\ge r$
  and $B(u_i)$ intersects $C$.

  \begin{figure}
    \begin{center}
      \includegraphics{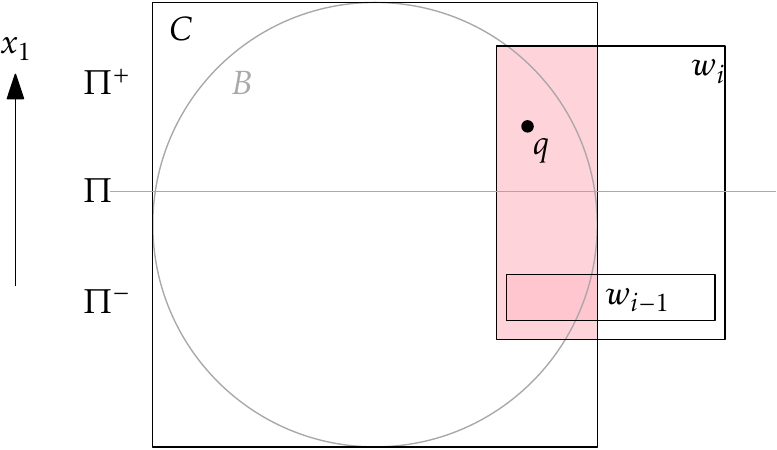}
    \end{center}
    \caption{Proving Property~5 of $k$-partitions.}
    \figlabel{property-5}
  \end{figure}

  For each $u\in U$, \eqref{bounds} implies that $L(u)\ge 2r/\sqrt{d}$.
  Therefore, by \lemref{box-packing}, $U$ contains a subset, $U'$, of
  size at most $(4\sqrt{d}+2)^d\in O(1)$ such that every node in $U$
  is an ancestor of some node in $U'$.  Thus, the elements of $U$ can be
  covered by $O(1)$ paths, each of which goes from a node in $U'$ to the 
  root of $T$.  It suffices to consider the contribution of one such
  path, $w_1,\ldots,w_\ell$, where $w_1\in U'$ and $w_\ell$ is the root
  of $T$.

  For each $i\in\{1,\ldots,\ell\}$, let $C_{i}$, denote the box $C\cap
  B(w_{i})$.  Since $B(w_1)\subset\cdots\subset B(w_\ell)$, we have that
  $C_1\subseteq\cdots\subseteq C_\ell$. Observe that, for each $w_i$, the
  ball associated with $w_i$ is not assigned any points in $B(w_{i-1})$.
  Thus, it is sufficient to show that there are $O(1)$ values of $i$
  for which $V\cap C_i\neq V\cap C_{i-1}$; for each such $i$, the number
  of elements assigned to the corresponding ball of the $k$-partition
  is at most $k$.

  We extend the side-length notation, $L_i$, to any box,
  $B=[a_1,b_1]\times\cdots[a_d,b_d]$, so that $L_i(B)=|b_i-a_i|$ and
  define the \emph{total side length} $\sum_{i=1}^d L_i(B)$.  We will
  show that, for each $i\in\{2,\ldots,\ell\}$, at least one of the
  following statements is true
  \begin{enumerate}
    \item $V\cap C_i = V\cap C_{i-1}$; 
    \item the total side length of $C_i$ exceeds that of $C_{i-1}$ by at least
      $L(w_1)/2$; or
    \item $C_i$ intersects a side of $C$ that is not intersected by $C_{i-1}$.
  \end{enumerate}
  This is sufficient to prove the result since Case~1 does not result in
  any new points included in $B$, Case~2 can occur at most $4rd/L(w_1)\in
  O(1)$ times, and Case~3 can occur at most $2d\in O(1)$ times.

  To see why one of the preceding cases must occur, suppose that neither
  Case~1 nor Case~3 applies.  Since Case~1 does not apply, there is
  some point $q\in V\cap C_i$ that is not in $C_{i-1}$. 
  Without loss of generality, assume that the fair-split tree cuts
  $B(w_i)$ with a plane, $\Pi$, that is perpendicular to the $x_1$-axis.
  Let $\Pi^+$ and $\Pi^-$ denote the closed halfspaces bounded by $\Pi$
  that contain $q$ and $B(w_{i-1})$, respectively.  Then we have that
  \[
      L_1(B(w_i)\cap \Pi^+) 
          = L_1(w_i)/2 
          \ge L(w_{i-1})/2 
          \ge L(w_{1})/2  \enspace .
  \]
  Observe that $B(w_{i})$ does not intersect the side of $C$ that is
  parallel to $\Pi$ and contained in $\Pi^+$ (since, otherwise, Case~3
  would apply).  This implies that
  \[
      L_1(C_i) \ge L_1(C_{i-1}) + L(w_1)/2 \enspace .
  \] 
  Thus, if neither Case~1 nor Case~3 applies to $u_i$, then Case~2
  applies.  This completes the proof.
\end{proof}

\subsection{The Graph $G$}

With the availability of $k$-partitions, we are now ready to construct
a graph $G$ with low average stretch factor.  In the following construction,
positive valued variables $c,k\in\omega_n(1)$ and $\epsilon\in o_n(1)$
are used without being specified.  Values of these variables that
optimize the average stretch factor of $G$ will be given in the proof
of \thmref{upper-bound}.  In the meantime, the reader can mentally assign the
values $c=k=\log n$ and $\epsilon = 1/\log n$, which are sufficient to
prove that $\asf(G)=1+o_n(1)$.

\paragraph{Hubs.}

We begin with a $k$-partition $(\{\Delta_1,\ldots,\Delta_{n'}\},f)$
of $V$.  For each $i\in\{1,\ldots,n'\}$, let $r_i$ denote the radius of
$\Delta_i$.  We will use the convention that $\Delta_1,\ldots,\Delta_{n'}$
are ordered by increasing radii, so that $r_i \le r_j$ for each $1\le
i < j\le n'$.  For each $i\in \{1,\ldots,n'\}$, let $V_i=\{u\in V :
f(u)=\Delta_i\}$; that is, $V_i$ is the set of points assigned to the
ball $\Delta_i$.  For each set $V_i$, we choose a \emph{hub}, $u_i\in
V_i$, arbitrarily. Let $H=\{u_1,\ldots,u_{n'}\}$ denote the set of hubs
and recall that $|H|=n'\in O(n/k)$.

\begin{figure}
  \begin{center} 
    \includegraphics{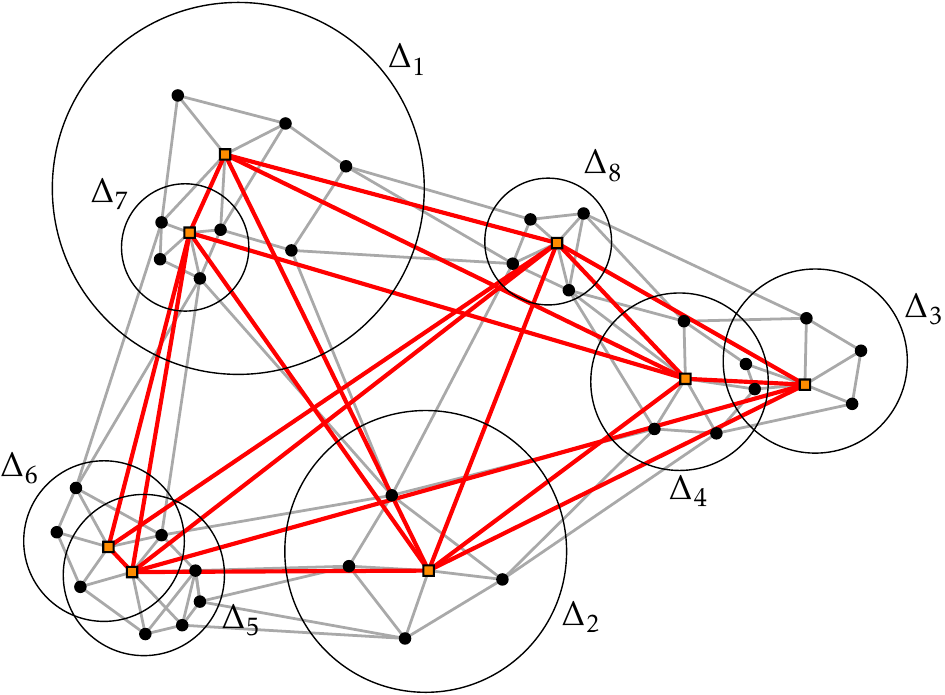}
  \end{center} 
  \caption{$G$ contains a 2-spanner of $V$ (gray edges) as well as
    $O(n/k)$ hubs whose centers are interconnected by a
    $(1+1/k^{1/(d-1)})$-spanner (red edges)}
  \figlabel{overview}
\end{figure}

\paragraph{Roads and Highways.}

Our graph $G$ starts with two spanner constructions.
The first spanner (the roads), denoted by
$G_2=(V,E_2)$, is a 2-spanner of $V$, and has $O(n)$ edges
\cite{callahan.kosaraju:faster,salowe:constructing,vaidya:sparse}.
The next spanner (the highways), denoted by $G_1=(H,E_1)$, is a
$(1+1/k^{1/(d-1)})$-spanner of $H$, and has $O(k|H|)=O(n)$ edges
\cite{carmi.smid:optimal,ruppert.seidel:approximating}.

With $G_1$ and $G_2$ we have, for any $u,w\in V$:
\[
   \frac{\|uw\|_G}{\|uw\|} \le \begin{cases}
         1+1/k^{1/(d-1)} & \text{if $u,w\in H$ (by using $G_1$)} \\
         2 & \text{in any case (by using $G_2$).}
       \end{cases}
\]

\paragraph{Covering a Nearby Cluster.}

Informally, the idea behind the graphs $G_1$ and $G_2$ is that, if
$u\in V_i$ and $w\in V_j$ are ``far apart'' (relative to $r_i$
and $r_j$), then the path from $u$ to $w$ that goes from $u$ to $u_i$
via roads ($G_2$), then to $u_j$ via highways ($G_1$), and then onto
$w$ via roads again should have length $(1+o_n(1))\|uw\|$.  Of course,
this only works if $u$ and $w$ are far apart.  The final step of
our construction attempts to deal with the majority of cases where $u$
and $w$ are not far apart.

To make the preceding ideas precise, let $D_i$ be the ball centered at
the center of $\Delta_i$ and having radius $cr_i$.  Points of $V$ that
are in $D_i$ can be problematic for $V_i$; there is no guarantee that such
points have paths with $1+o_n(1)$ stretch factor to the points in $V_i$.

For each $i\in\{1,\ldots,n'\}$, we find a ball, $E_i$, of radius $r_i/c$,
that intersects $D_i$, and that contains the maximum number of points
of $V$.  (Note that this may include points of $V$ in $V_i$ or outside
of $D_i$.)  We then add edges joining each of the points in $V_i$ to
a carefully chosen point $w_i\in E_i$.  See \figref{hitter}.

The point $w_i$ is chosen as follows: For each point $w\in V$, let
$i(w)\in \{1,\ldots,n'\}$ denote the smallest index such that $w\in
E_{i(w)}$ and $|E_{i(w)}\cap V| \ge \epsilon n$; if no such index exists,
let $i(w)=\infty$.  The point $w_i\in E_i$ is selected to be any of the
points in $E_i$ that minimizes $i(w)$.  This concludes the description
of the graph $G$.

\begin{figure}
  \begin{center}
    \includegraphics{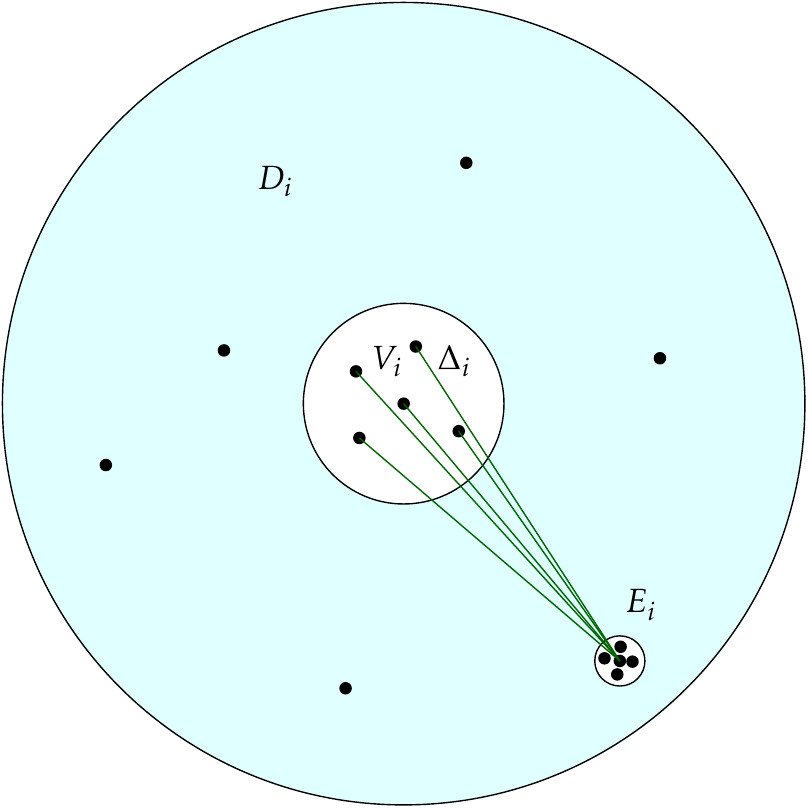}
  \end{center}
  \caption{The ball $E_i$ captures as many points of $V$ as possible
   while still intersecting $D_i$.}
  \figlabel{hitter}
\end{figure}

\subsection{Two Illustrative Examples}
Before delving into the proof that $G$ has low average stretch factor,
it may be helpful to study two examples that illustrate why the balls
$E_1,\ldots,E_{n'}$ are needed and why the choice of the representative
vertices, $w_i\in E_i$, is  important.

\paragraph{Example 1: Exponential Grids.}
The first example is a set of points arranged as a sequence of
$\sqrt{k}\times\sqrt{k}$ grids, $G_0,\ldots,G_{n/k-1}$.  The grid $G_i$
has its center on the x-axis at x-coordinate $2^{i+1}-1$, and has side
length $2^i$ (see \figref{grid}). The natural $k$-partition of this
grid is the one that assigns all points in each $G_i$ to a single ball,
$\Delta_i$.  In this grid, if we consider $G_i$, for some large value
of $i$, we see that all the points in $G_0,\ldots,G_{i-1}$ are within
distance $O(2^{i})$ of all the points in $G_i$.  Forcing every path from
every $u\in G_i$ to every $w\in G_{j}$, $j<i$, to go through a central
vertex $u_i\in G_i$ would be too costly; on average the detour through
$u_i$ would increase the length of this path by $\Omega(2^{i})$.

\begin{figure}
  \begin{center}
    \includegraphics{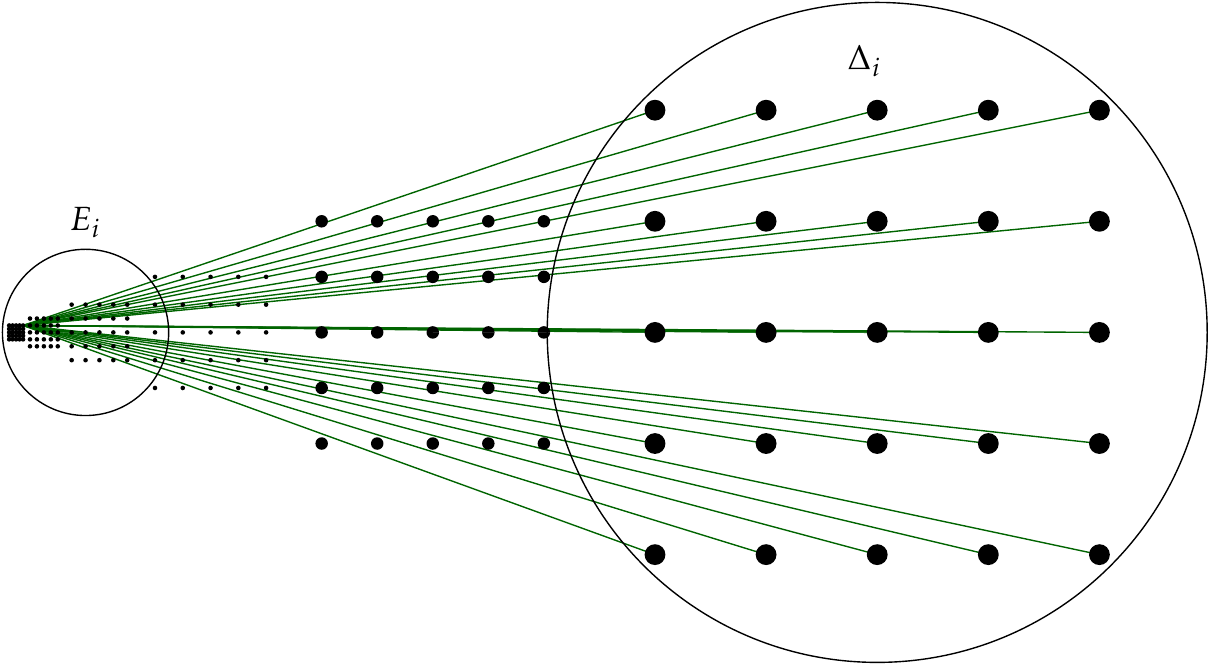}
  \end{center}
  \caption{A sequence of exponentially increasing grids illustrates the
   need for connecting all points in $V_i$ to some point in $E_i$.}
  \figlabel{grid}
\end{figure}

The ball $E_i$ solves the preceding problem; $E_i$ is large enough
to cover all points in $G_0,\ldots,G_{i-\Theta(\log c)}$.  The path
from $u\in G_i$ directly to $w_i\in E_i$ and then to any $w\in E_i$
has length at most
\[
    \|uw\| + O(2^{i}/c) \enspace .
\]
Furthermore, all of the points in $E_i$ are at distance $\Omega(2^i)$ from
all the points in $G_i$, so $\|uw\|_G/\|uw\| = 1+O(1/c)$.  Part~3 of the
proof of \thmref{upper-bound} is dedicated to showing that, in general,
the balls $E_1,\ldots,E_{n'}$ help with many pairs of points that would
otherwise be problematic.

\paragraph{Example 2: The Importance of Choosing Wisely.}
Our second example is intended to illustrate the importance of carefully
choosing the representative vertex $w_i\in E_i$.   In this example, there
is a dense cluster of $n/2$ points that is contained in some ball $E_j$
(see \figref{example2}).  Consider now some $i>j$ such that $r_i$ is
much greater than $r_j$.  It is easy to make a configuration of points so
that, for some cluster $V_i$, the corresponding ball $E_i$ contains $E_j$
as well as a few other points that are far from $E_j$. If one of these
points is used as the representative vertex, $w_i$, then all $kn/2$
paths from some $u$ in $V_i$ to some $w\in E_j$ will have to make a
detour through $w_i$.  By repeating this for many different values of
$i$, this is enough to produce an average stretch factor significantly
larger than 1.

\begin{figure}
  \begin{center}
    \includegraphics{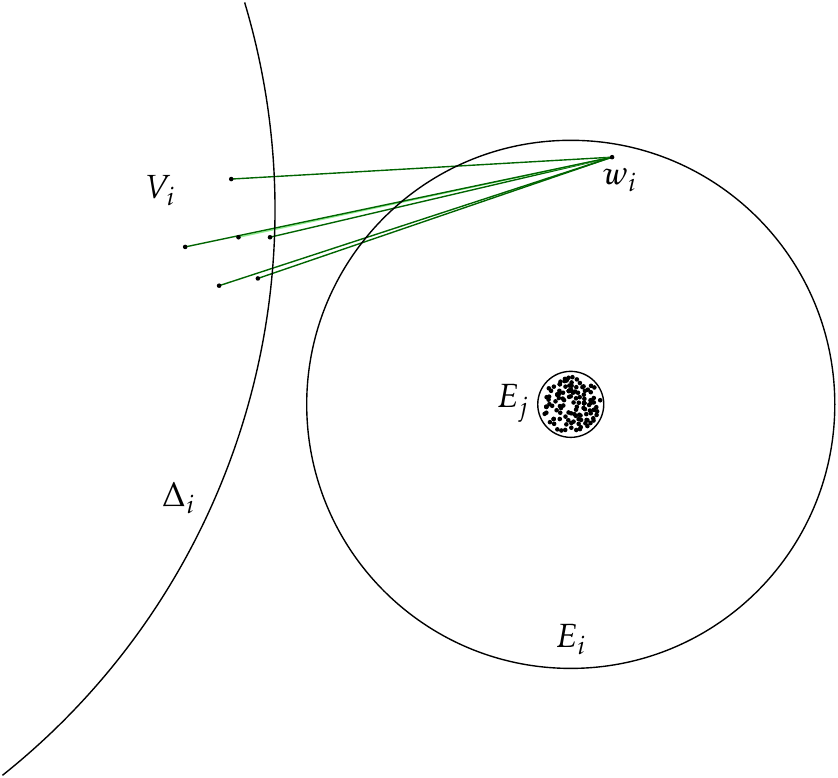}
  \end{center}
  \caption{An illustration of why it is important to choose $w_i$ carefully.
    A bad choice (like the one illustrated) leads to a significant detour
    on the paths from every $u\in V_i$ to every $w\in E_j$.}
  \figlabel{example2}
\end{figure}

The choice of $w_i$ is designed to avoid the preceding problem.  In this
example, $w_i$ would be chosen from the points in $E_j$, since $r_j <
r_i$ and $E_i$ contains $n/2\ge \epsilon n$ points of $V$.  Part~4 of
the proof of \thmref{upper-bound} is dedicated to showing this careful
choice of $w_i$ works.  In particular, it ensures that, for most pairs
of the form $u\in V_i$, $w\in E_i$,
\[
    \|uw_i\| + O(\|w_iw\|) = \|uw\|(1+O(1/c)) \enspace .
\]

\subsection{The Proof}

Without further ado, we now prove that $G$ has low average stretch factor.

\begin{thm}\thmlabel{upper-bound}
  For every constant dimension, $d$, and every set, $V$, of
  $n<\infty$ points in $\R^d$, the graph $G=(V,E)$ described above
  has $O(n)$ edges and $\asf(G)=1+o_n(1)$.  More precisely,
  $\asf(G)=1+O((\log n/n)^{1/(2d+1)})$.
\end{thm}

\begin{proof}
  That $G$ has $O(n)$ edges follows immediately from its definition.

  To upper-bound the average stretch factor of $G$, there are four types
  of pairs of points, $u\in V_i$, $w\in V_j$, $j\le i$, to consider
  (recall that $\Delta_1,\ldots,\Delta_{n'}$ are ordered so that $r_j \le r_i$):
  \begin{enumerate}
    \item pairs that are both from the same set, i.e., where $i=j$;
    \item pairs for which $w$ is outside of $D_i$;
    \item pairs for which $w$ is contained in $D_i\setminus E_i$; and
    \item pairs for which $w$ is contained in $E_i\cap D_i$.
  \end{enumerate}
  We consider each of these types of pairs in turn.  Our strategy is to
  study the $\binom{n}{2}$ terms that define $\asf(G)$ in \eqref{asf}.
  We will show that $o(n^2)$ of these terms are at most 2 while the
  remaining terms are at most $1+o_n(1)$.  Thus,
  \[
     \asf(G)\le \binom{n}{2}^{-1}\left(2\cdot o(n^2)
                                       +\binom{n}{2}(1+o_n(1))\right)
     = 1+o_n(1) \enspace .
  \]

  \paragraph{Type~1 Pairs.}
  Each $V_i$, for $i\in\{1,\ldots,n'\}$, defines at most $\binom{k}{2}$
  Type~1 pairs, so the total number of Type~1 pairs that contribute a
  term to the sum in \eqref{asf} is at most
  \[
    \binom{k}{2}\cdot O(n/k) \in O(nk)
      \enspace .
  \]

  \paragraph{Type~2 Pairs.}
  For each Type~2 pair $u\in V_i$, $w\in V\setminus D_i$, there is a
  path from $u$ to $u_i$ in $G_2$, then from $u_i$ to $u_j$ in $G_1$
  and then finally from $u_j$ to $w$ in $G_2$ whose length is at most
  \[
     4r_i + \|u_iu_j\|_G + 4r_j
      \le (1+1/k^{1/(d-1)})(\|uw\| + 16r_i) \enspace .
  \]
  Furthermore, $\|uw\|\ge (c-1)r_i$ since $w$ is outside of $D_i$.
  Therefore, for each Type~2 pair, the term that appears in \eqref{asf}
  is of the form
  \[
    \frac{\|uw\|_G}{\|uw\|}\le (1+1/k^{1/(d-1)})(1+O(1/c)) 
       = 1+O(1/k^{1/(d-1)}+1/c) \enspace .
  \]

  \paragraph{Type~3 Pairs.}
  The number of Type~3 pairs is no more than 
  \[  
     k\cdot\sum_{i=1}^{n'}|V\cap D_i\setminus E_i| \enspace .
  \]
  We will prove, by contradiction, that this quantity is $o(n^2)$.
  Suppose, for the sake of contradiction, that this is not the case
  and that
  \begin{equation}
    \sum_{i=1}^{n'} |V\cap D_i\setminus E_i| \ge \frac{\delta n^2}{k}
         \enspace , \eqlabel{kicker}
  \end{equation}
  where $\delta>0$ will be determined later.
  Each term on the left hand side of \eqref{kicker} is at most $n$
  and there are $n'\le \alpha n/k$ terms, for some constant $\alpha
  >0$.  We say that a term on the left hand side of \eqref{kicker} is
  \emph{small} if it is less than $\delta n/2\alpha$ and \emph{large}
  otherwise.  The sum of the small terms is at most $\delta n^2/2k$
  and therefore the sum of the large terms is at least $\delta n^2/2k$.
  Let $J$ be the index set of these large terms.  Then
  \[
    \sum_{i\in J} |V\cap D_i\setminus E_i| \ge \frac{\delta n^2}{2k} \enspace .
  \]
  By the pigeonhole principle, there must exist some point $w^*\in V$
  such that there are at least $\delta n/2k$ indices $i\in J$ such that
  $w^*\in V\cap D_i\setminus E_i$.  To summarize the discussion so far:
  There exists a point $w^*\in V$ and index set, $I\subseteq J$,
  such that:
  \begin{enumerate}
     \item[A1.] $w^*\in V\cap D_{i}\setminus E_{i}$, for all
        $i\in I$
     \item[A2.] $|V\cap D_{i}\setminus E_{i}|\in \Omega(\delta n)$,
       for all $i\in I$; and
     \item[A3.] $|I|\in\Omega(\delta n/k)$.
  \end{enumerate}

  Suppose, without loss of generality, that the smallest ball,
  $\Delta_i$, with $i\in I$ has unit radius.  Partition $I$ into groups
  $G_0,G_1,\ldots$ such that $G_t$ contains all indices $i\in I$ such
  that $\Delta_{i}$ has radius in the interval $[2^t,2^{t+1})$.  We claim
  that each such group, $G_t$, has size $O(c^d)$.  To see why this is
  so, observe that, for each group $G_t$, there exists a point---namely
  $w^*$---that is contained in $|G_t|$ balls $D_{i}$ where $i\in G_t$.
  $D_{i}$ has radius at most $c2^{t+1}$.  This means that the set of balls
  \[
     \{ \Delta_i : i\in G_t\}
  \]
  is contained in a ball, centered at $w^*$, of radius at most
  $(c+2)2^{t+1}$.  Since each ball in this set has radius in
  $[2^t,2^{t+1})$, a standard packing argument implies that some point
  $p\in\R^d$ is contained in $\Omega(|G_t|/c^d)$ of these balls. But then
  Property~4 of $k$-partitions implies that the size of $|G_t|/c^d\in
  O(1)$, so $|G_t|\in O(c^d)$.

  Thus far, we have shown that each group, $G_t$, has size $O(c^d)$
  and the total size of all groups is $|I|\in\Omega(\delta n/k)$.
  Therefore, there must be at least $\Omega(\delta n/(kc^d))$ groups.
  In particular, we can find $h\in\Omega(\delta n/(kc^d\log c))$ groups,
  $G_{t_1},\ldots,G_{t_h}$, such that $t_{i+1} \ge t_{i}+2\log c+2$ for
  each $i\in\{1,\ldots,h-1\}$.  By selecting a representative element from
  each of these groups, we obtain a sequence of indices $i_1,\ldots,i_h$
  such that the radius of $\Delta_{i_{j+1}}$ is at least $4c^d$ times
  the radius of $\Delta_{i_j}$ for each $j\in\{1,\ldots,h-1\}$.

  By choice, $D_{i_1}$ contains at least $ a \delta n$ elements of
  $V$, for the constant $a=1/2\alpha$.  Also by choice, $D_{i_2}\setminus
  E_{i_2}$ contains at least $ a \delta n$ elements of $V$.
  Both $D_{i_1}$ and $D_{i_2}$ contain the point $w^*$.  We claim
  that $E_{i_2}$ contains at least $ a \delta n$ elements of $V$
  as well since there exists a ball, $E_{i_2}'$, centered at $w^*$,
  of radius $r_{i_2}/c > 4cr_{i_1}$, that contains $D_{i_1}$ and
  therefore contains all the (at least $ a\delta n$) points in $D_{i_1}$
  (see \figref{containment}).  The ball $E_{i_2}$ was chosen to contain
  as many elements of $V$ as possible, so it contains at least as many
  elements as $E_{i_2}'$.  Therefore, $D_{i_2}\cup E_{i_2}$ contains at
  least $2 a  \delta n$ points of $V$.

  \begin{figure}
     \begin{center}
       \includegraphics{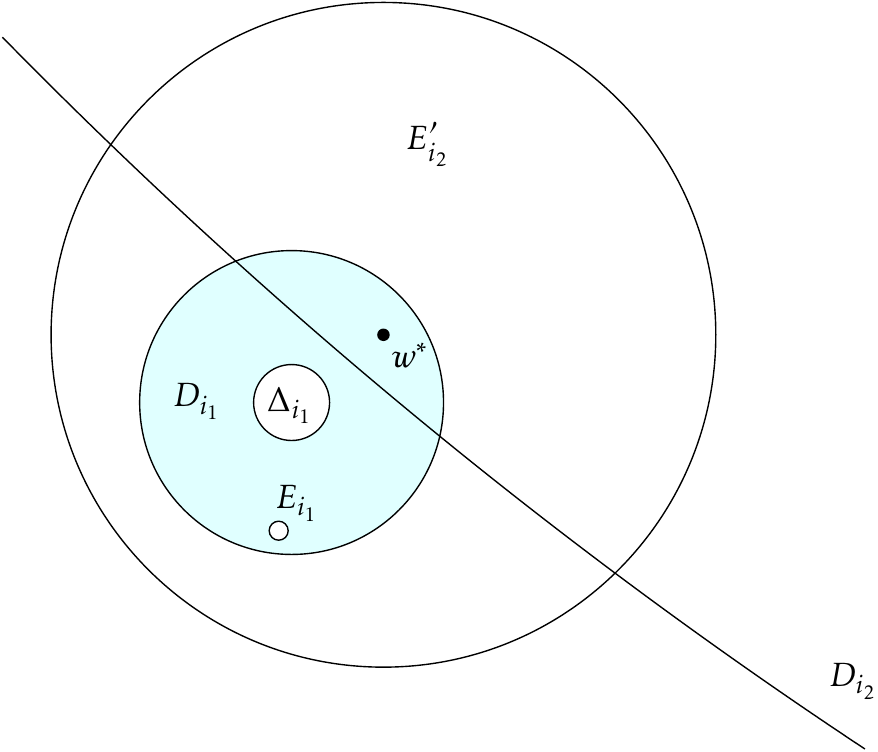}
     \end{center}
     \caption{$D_{i_2}\cup E_{i_2}$ contains at least $2 a  n$ 
              points of $V$.}
     \figlabel{containment}
   \end{figure}

  We can now argue similarly to show that $E_{i_3}$ contains at
  least $2 a \delta n$ points of $V$ so $D_{i_3}\cup E_{i_3}$
  contains at least $3 a \delta n$ points of $V$.  In general,
  this argument shows that $D_{i_h}\cup E_{i_h}$ contains at least
  $h a \delta n$ points of $V$.  But this yields a contradiction for
  $h> 1/ a \delta$, since $V$ contains only $n$ points.  To obtain
  this contradiction, our choice of $\delta$, $c$, and $k$ must satisfy
  \[
       h\in\Omega\left(\frac{\delta n}{kc^d\log c}\right) \ge
          \frac{1}{ a  \delta}
  \]
  which is satisfied by any choice of $\delta$, $c$, and $k$ such that
  \[
       \frac{\delta^2 n}{kc^d\log c} \ge \frac{C}{a}
  \]
  for some sufficiently large constant $C$.  In particular, the value
  \[
       \delta = \sqrt{\frac{Ckc^d\log c}{an}}
  \]
  works.  So the total number of terms of the sum in \eqref{asf}
  contributed by Type~3 pairs is at most
  \[
    \delta n^2 \in O\left(\sqrt{kn^3c^d\log c}\right) \enspace .
  \]

  \paragraph{Type~4 Pairs.}  
  Let $\beta > 0$ be a constant whose value will be discussed later.
  We say that a Type~4 pair of points $u\in V_i$, $w\in E_i\cap D_i$ is
  a \emph{bad pair} if
  \[
      \|uw_i\|+2\|w_iw\| \ge (1+\beta/c)\|uw\| \enspace ,
  \]
  and otherwise it is a \emph{good pair}.  For any good pair $(u,w)$,
  \[
    \frac{\|uw\|_G}{\|uw\|} = 1+O(1/c) \enspace ,
  \]
  so we can focus our study on bad pairs.  To upper-bound the number
  of bad pairs we will assume, for the sake of contradiction, that the
  number of bad pairs is at least $\epsilon n^2$.

  Let $b_i$ denote the number of bad pairs $(u,w)$ with $u\in V_i$
  and $w\in E_i\cap D_i$.  Then, by assumption,
  \[
    \sum_{i=1}^{n'} b_i \ge \epsilon n^2 \enspace .
  \]
  Applying the same reasoning used to study Type~3 pairs, we can find a
  point $w^*\in V$ and a set of indices $i_1,\ldots,i_{\ell}$ such that
  \begin{enumerate}\label{w-star}
    \item[B1.] $w^*\in E_{i_j}$, for all $j\in\{1,\ldots,\ell\}$;
    \item[B2.] $b_{i_j} \in \Omega(\epsilon kn)$, for all 
       $j\in\{1,\ldots,\ell\}$; and
    \item[B3.] $\ell\in \Omega(\epsilon n/k)$.
  \end{enumerate}

  Assume that the indices $i_1,\ldots,i_\ell$ are ordered so that
  $r_{i_j}\le r_{i_{j+1}}$ for all $j\in\{1,\ldots,\ell-1\}$.  Consider
  the sequences $E'_{i_1},\ldots,E'_{i_\ell}$, where each $E'_{i_j}$
  is the ball of radius $2r_{i_j}/c$ centered at $w^*$. Recall that
  the radius of $E_{i_j}$ is $r_{i_j}/c$ and $w^*\in E_{i_j}$, so that
  $E'_{i_j}\supset E_{i_j}$, for each $j\in\{1,\ldots,\ell\}$.

  The plan for the rest of the proof is as follows:  We will find an
  annulus $A=E'_{i_{j^*+t+C}}\setminus E'_{i_{j^*}}$ that does not contain
  very many points of $V$.  We will then use the fact that $A$ does not
  contain many of points of $V$ and Property~5 of $k$-partitions to prove
  that, for some index $j\in\{j^*,\ldots,j^*+t\}$, $b_{i_j}< D\epsilon kn$
  for any constant $D>0$.  This yields the desired contradiction, since
  $i_1,\ldots,i_\ell$ were chosen so that $b_{i_j}\in\Omega(\epsilon k n)$
  for every $j\in\{1,\ldots,\ell\}$.

  To begin, we fix some positive integers $C\in O(1)$ and $t< \ell
  - C$ to be specified later.  For each $j\in\{2,\ldots,\ell\}$,
  let $n_{i_j}=|E'_{i_j}\cap V\setminus E'_{i_{j-1}}|$. We have that
  $\sum_{i=2}^{\ell} n_{i_j} \le n$ and $\ell \in\Omega(\epsilon n/k)$.
  Using these two bounds, a simple averaging argument is sufficient
  to show that there must exist an index $j^*\in\{1,\ldots,\ell-t-C\}$
  such that
  \begin{equation}
     \sum_{j=j^*+1}^{j^*+t+C} n_{i_j}
        = |V\cap E'_{i_{j^*+t+C}}\setminus \E'_{i_{j^*}}| 
          \in O((t+C)k/\epsilon)
          = O(tk/\epsilon) \enspace . \eqlabel{sparse}
  \end{equation}

  The careful choice of each $w_i\in E_i$ implies the following claim,
  whose proof is deferred until later.
  \begin{clm}\clmlabel{zuper}
    For every $j\in\{j^*+1,\ldots,j^*+t\}$, every $u\in V_{i_j}$, and
    every $w\in E_{i_j}\cap E'_{i_{j^*}}\cap V$, $G$ contains a path of
    length at most $\|uw\|+O(r_{i_{j^*}}/c)$.
  \end{clm}

  \begin{figure}
    \begin{center}
      \includegraphics{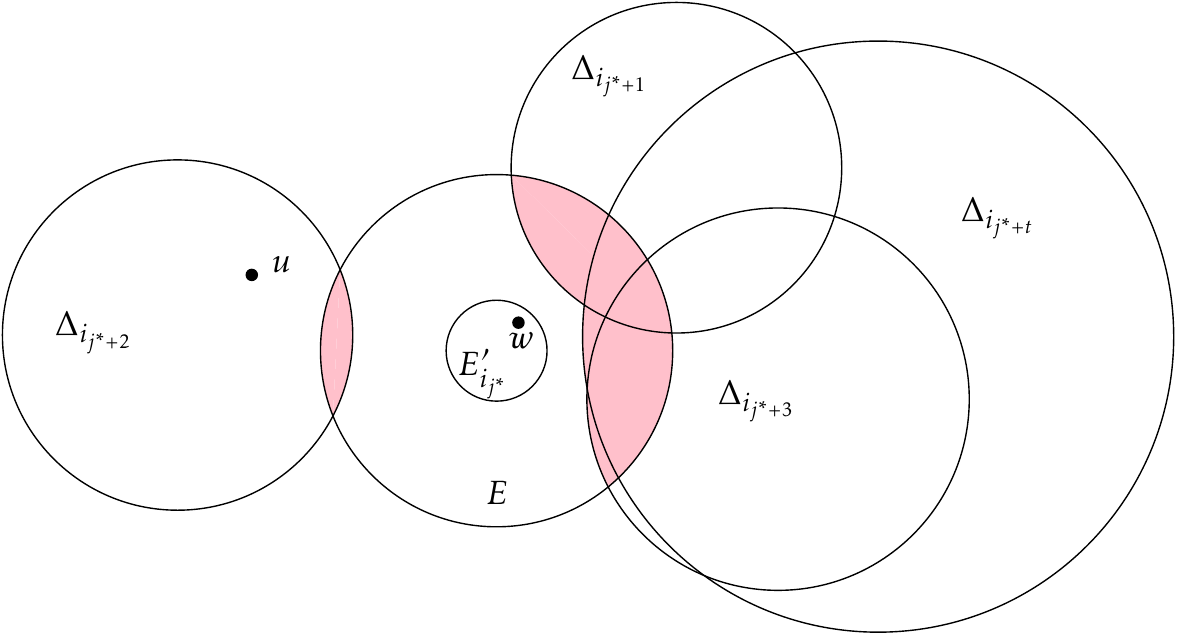}
    \end{center}
    \caption{The number of points in $E$ that are assigned to
      $\Delta_{i_{j^*}},\ldots,\Delta_{i_{j^*+t}}$
      is only $O(k)$.}
    \figlabel{type4}
  \end{figure}

  Refer to \figref{type4}.  Let $E$ denote the ball centered at $w^*$
  and having radius $r_{i_{j^*}}=c\cdot\radius(E_{i_{j^*}})$.  Note that
  any point $u\not\in E$ is at distance at least $(1-2/c)r_{i_{j^*}}$ from
  any point $w\in E'_{i_{j^*}}$.  Therefore, by \clmref{zuper}, for any
  $w\in E_{i_j}\cap E'_{i_{j^*}}\cap V$, any $j\in\{j^*+1,\ldots,j^*+t\}$
  and any $u\in V_{i_{j}}\setminus E$,
  \[  
     \frac{\|uw\|_G}{\|uw\|} = 1+O(1/c) \enspace . 
  \]
  In other words, by choosing an appropriate constant $\beta$ in the
  definition of bad pairs, $u$ can not form a bad pair with a point
  $w\in E'_{i_{j^*}}$ unless $u$ is contained in $E$.


  Next we will upper-bound $\sum_{j=j^*+1}^{j^*+t}b_{i_j}$, the number
  of bad pairs, $(u,w)$, with $u\in V_{i_j}$, $w\in E_{i_j}$, and
  $j\in\{j^*+1,\ldots,j^*+t\}$.  From the preceding discussion, each such
  pair falls into at least one of the two following categories:
  \begin{enumerate}
    \item Category A: $u\in E$. 
      To bound the number of these pairs, consider the balls
      $\Delta_{i_{j^*+1}},\ldots,\Delta_{i_{j^*+t}}$. Each of these
      balls has radius at least $r_{i_{j^*}}=\radius(E)$.  Therefore,
      Property~5 of $k$-partitions implies that
      \[
        \left|\bigcup_{j=j^*+1}^{j^*+t} V_{i_{j}}\cap E\right|
          \in O(k)
      \]
      Each $u$ in $V_{i_j}\cap E$ takes part in at most $O(n)$ bad pairs,
      so the number of bad pairs of this form is $O(kn)$.

    \item Category B: $w\not\in E'_{i_{j^*}}$.
      By \eqref{sparse}, the number of points that are not in
      $E'_{i_{j^*}}$ but still in some $E_{i_j}\subset E'_{i_j}$
      for $j\in\{j^*+1,\ldots,j^*+t\}$ is $O(tk/\epsilon)$.
      Each such point, $w$, forms a bad pair with at most
      $tk$ different values of $u$---namely the elements of
      $V_{i_{j^*}+1},\ldots,V_{i_{j^*+t}}$. Therefore, the number of
      bad pairs in this category is at most $O((tk)^2/\epsilon)$.
  \end{enumerate}
  Thus, we have 
  \begin{equation}
    \sum_{j=j^*+1}^{j^*+t} b_{i_{j}}  \in O(kn + (tk)^2/\epsilon) \enspace .
    \eqlabel{ub-4}
  \end{equation}
  On the other hand, the indices $i_1,\ldots,i_\ell$ are chosen
  so that $b_{i_j}\in\Omega(\epsilon kn)$, for \emph{every}
  $j\in\{1,\ldots,\ell\}$.  Therefore
  \begin{equation}
    \sum_{j=j^*+1}^{j^*+t} b_{i_{j}} 
        \in \Omega(t\epsilon kn) \enspace .
        \eqlabel{lb-4}
  \end{equation}
 
  Equating the right hand sides of \eqref{ub-4} and \eqref{lb-4}, we
  obtain 
  \[
     C_1t\epsilon kn \le C_2(kn + (tk)^2/\epsilon) \enspace , 
  \]
  for some constants $C_1$ and $C_2$.  This yields a contradiction when 
  \[
      t > D/\epsilon
  \]
  and
  \[
          \epsilon\ge D\sqrt{\frac{tk}{n}} \enspace ,
  \]
  for some sufficiently large constant $D$.  Setting $t=D/\epsilon$
  leaves only the condition
  \[
          \epsilon \ge D^{2/3}\left(\frac{k}{n}\right)^{1/3} \enspace .
  \]
  Thus, there exists some constant $D_0$ such that the total number of
  bad Type~4 pairs is at most
  \[
    \epsilon n^2 \enspace ,
  \]
  provided that $\epsilon \ge D_0(k/n)^{1/3}$. 

  All that remains in handling Type~4 pairs is to prove \clmref{zuper}.

  \begin{proof}[Proof of \clmref{zuper}]
  Let $u$ by any point in $V_{i_{j}}$, for any
  $j\in\{j^*+1,\ldots,j^*+t\}$.  Since there is an edge joining $u$ and
  $w_{i_{j}}$ and $\|w^*w\|\le 2r_{i_{j^*}}/c$, there is a path from $u$
  to $w$ of length at most
  \begin{align*}
      \|uw\|_G & \le \|uw_{i_j}\| + 2\|w_{i_j}w^*\| + 2\|w^*w\|  \\
         & \le \|uw\| + \|ww_{i_j}\| + 2\|w_{i_j}w^*\| + 4r_{i_j^*}/c  \\
         & \le \|uw\| + \|ww^*\| + 3\|w_{i_j}w^*\| + 4r_{i_j^*}/c \\
         & \le \|uw\| + 3\|w_{i_j}w^*\| + 6r_{i_j^*}/c 
     \enspace ,
  \end{align*}
  so it is sufficient to prove that $\|w^*w_{i_{j}}\|\in
  O(r_{i_{j^*}}/c)$.  If $w_{i_{j}}=w^*$, then we are done, so assume
  $w_{i_{j}}\neq w^*$.  Since $w^*\in E_{i_{j}}$ but was not chosen
  to act as $w_{i_{j}}$, there must exist some index $i'$ such that
  $E_{i'}\cap E_{i_{j}}\neq\emptyset$, $|E_{i'}\cap V|\ge\epsilon n$,
  $r_{i'}\le r_{i_{j^*}}$, and $w_{i_{j}}\in E_{i'}$. (Note for later:
  This is the only place in the entire proof of \thmref{upper-bound}
  where the choice of $w_{i_j}$ matters.)

  \noindent There are two cases to consider:
  \begin{enumerate}
    \item 
    If $E'_{i_{j^*}}$ intersects $E_{i'}$, then (see \figref{gummo-a})
    \[
       \|w^*w_{i_j}\|\le 2r_{i_{j^*}}/c + 2r_{i'}/c \le 4r_{i_{j^*}}/c \in O(r_{i_{j^*}}/c) \enspace ,
    \]
    and we are done. 
    \begin{figure}
      \begin{center}
        \includegraphics{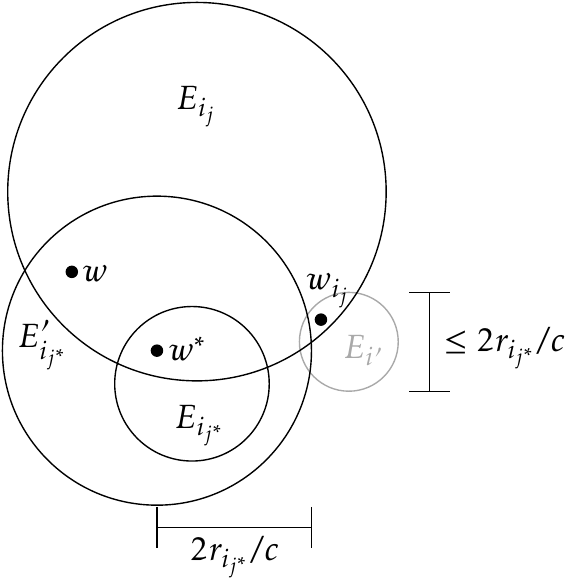}
        \caption{If $E_{i'}$ intersects $E'_{i_j^*}$ then $\|w^*w_{i_j}\|\le 4r_{i_{j^*}}/c$.}
        \figlabel{gummo-a}
      \end{center}
    \end{figure}
  
    \item If $E_{i'}$ and $E'_{i_{j^*}}$ are disjoint, then we claim that
    $E_{i'}\subset E'_{i_{j^*+t+C}}$ (see \figref{gummo-b}).  To see why
    this is so, we argue that, when $C$ is a sufficiently large constant,
    \begin{equation}
      r_{i_{j^*+t+C}} > 2r_{i_{j^*+t}} 
          \ge 2r_{i_j} \enspace . \eqlabel{xxxx}
    \end{equation}
    This implies that $E_{i'}\subset E'_{i_{j^*+t+C}}$ since every
    point in $E_{i'}$ is at distance at most $4r_{i_j}$ from $w^*$
    and $E'_{i_{j^*+t+C}}$ is centered at $w^*$ and has radius
    $2r_{i_{j^*+t+C}} > 4r_{i_j}$.

    To see why the first inequality in \eqref{xxxx} holds, we first
    observe that, for any $i\in\{1,\ldots,n'\}$, if any $u\in V_i$
    and $w\in E_i$ form a bad pair, then the distance from $\Delta_i$
    to $E_i$ is less than $r_i$.  Since each $V_{i_j}$ and $E_{i_j}$
    define at least one pair, this implies that each $\Delta_{i_j}$ is
    contained in a ball of radius $(3+2/c)r_{i_j}$ centered at $w^*$.

    Now, if $r_{i_{j^*+t+C}} < 2r_{i_{j^*+t}}$, then
    $\Delta_{i_{j^*+t}},\ldots,\Delta_{i_{i_{j^*+t+C}}}$ is a set of
    $C+1$ balls all having radii in $[r_{i_{j^*+t}},2r_{i_{j^*+t}})$ and
    that are all contained in a ball of radius $(6+4/c)r_{i_{j^*+t}}$
    centered at $w^*$.  Therefore, some point, $p$, in this ball is
    contained in $\Omega(C)$ of these balls.  But then Property~4 of
    $k$-partitions implies that $C\in O(1)$.  Thus, for a sufficiently
    large constant, $C$, $r_{i_{j^*+t+C}} > 2r_{i_{j^*+t}}$ and
    $E_{i'}\subset E_{i_{j^*+t+C}}$, as promised.

    Since $E_{i'}$ and $E_{i_{j^*}}$ are disjoint and $E_{i'}\subset
    E_{i_{j^*+t+C}}$,
    \[
      |V\cap E_{i_{j^*+t+C}}\setminus E_{i_{j^*}}| 
         \ge |V\cap\E_{i'}| \ge \epsilon n \enspace .
    \]
    \begin{figure}
      \begin{center}
        \includegraphics{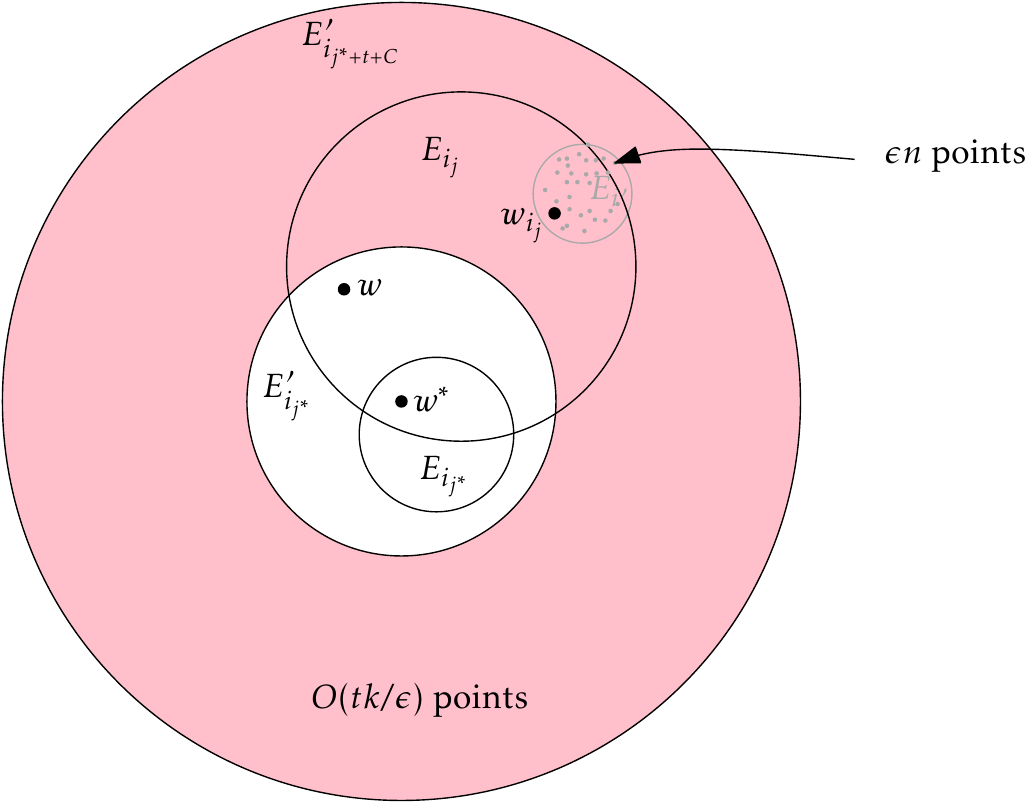}
        \caption{If $E_{i'}$ does not intersect $E'_{i_j^*}$ then
          $E_{i_{j^*+t+C}}\setminus E_{i_{j^*}}$ contains at least
          $\epsilon n$ points.}
        \figlabel{gummo-b}
      \end{center}
    \end{figure}
    But, by definition, 
    \[
      |V\cap E_{i_{j^*+t+2\mu}}\setminus E_i| 
         \in O(tk/\epsilon) \enspace .
    \]
    This yields a contradiction when $t<D' \epsilon^2 n/k$, for a
    sufficiently small constant $D'>0$.  Note that we have already set
    $t=D/\epsilon$, so this requires that
    \[
       D/\epsilon < D'\epsilon ^2 n/k \Leftrightarrow \epsilon 
           > \left(\frac{Dk}{D'n}\right)^{1/3} \enspace .
    \]
    Therefore, the condition on $t$ is already contained in the
    requirement that $\epsilon \ge D_0(k/n)^{1/3}$.
  \end{enumerate}
  This completes the proof of \clmref{zuper}.
  \end{proof}

  \paragraph{Finishing Up.}
  We can now pull everything together to summarize and complete the proof
  of \thmref{upper-bound}.
  \begin{enumerate}
    \item The number of Type~1 pairs is at most $O(kn)$.
    \item For each Type~2 pair, $(u,w)$, 
      $\|uw\|_G/\|uw\|\le 1+ O(1/k^{1/(d-1)}+1/c)$.
    \item The number of Type~3 pairs is at most
      $O(\sqrt{kn^{3}c^{d}\log c})$.
    \item For each good Type~4 pair, $(u,w)$, 
      $\|uw\|_G/\|uw\|\le 1+ O(1/c)$.
    \item The number of bad Type~4 pairs is at most 
       $O(\epsilon n^2)$ provided that $\epsilon \in \Omega((k/n)^{1/3})$.
  \end{enumerate}
  Putting all these together, we obtain
  \begin{equation}\eqlabel{master-asf}
     \asf(G) = 1 + O\left(1/k^{1/(d-1)} + 1/c 
       + \binom{n}{2}^{-1}\left(kn+\sqrt{kn^{3}c^{d}\log c}
             +\epsilon n^2\right)\right) \enspace .
  \end{equation}
  Taking 
  \[ 
       \epsilon = (k/n)^{1/3} \enspace ,
  \]
  \[ 
       k = \left(\frac{n}{\log n}\right)^{(d-1)/(2d+1)}
  \]
  and
  \[
       c = \left(\frac{n}{\log n}\right)^{1/(2d+1)}
  \]
  yields
  \[
     \asf(G) = 1+ O\left(\left(\frac{\log n}{n}\right)^{1/(2d+1)}\right) \enspace .
  \]
  (The $O(kn/\binom{n}{2})$ term and the $\epsilon n^2/\binom{n}{2}$
  term are dominated by the three other terms, which are all $O((\log
  n/n)^{1/(2d+1)})$.)
\end{proof}

\begin{rem}\remlabel{epsilon-choice}
   In the next section, efficient algorithms will require that
   $\epsilon = (\log^\kappa n)/k$ for some constant $\kappa \ge 0$.
   Since the bound on the average stretch factor requires that
   $\epsilon\in\Omega((k/n)^{1/3})$, this implies that the value of $k$
   must be at most $n^{1/4}$.  The following table shows the optimal
   choices of $k$ and $c$ for different dimensions, $d$:

   \begin{center}
   \renewcommand\arraystretch{1.3} 
     \begin{tabular}{|r|l|l|l|}
        \hline
        $d$ & $k$ & $c$ & $\asf(G)$ \\\hline\hline
        $2$ & $(n/\log n)^{1/5}$ & $(n/\log n)^{1/5}$ 
              & $1+ O((\log^{5\kappa+1} n/n)^{1/5})$ \\
        $\ge 3$ & $n^{1/4}$ & $(\log n/n^{3/4})^{1/(d+2)}$ 
              & $1+ O(1/n^{1/(4d-4)})$ \\
     \hline
     \end{tabular}
   \end{center}
\end{rem}

\section{Algorithms}
\seclabel{algorithms}

In this section, we discuss efficient algorithms for computing a graph
$G=(V,E)$ that has low average stretch factor, given the point set
$V\subset\R^d$.  We first present a fairly straightforward adaptation of
the construction given in the preceding section that yields an $O(n\log^d
n)$ time algorithm.  We then show that some refinements of this algorithm
lead to an $O(n\log n)$ time algorithm.

Our algorithms are randomized.  Throughout this section, we say that an
event within an algorithm happens \emph{with high probability} if the
event occurs with probability at least $1-O(n^{\alpha})$ and $\alpha$
can be made into an arbitrarily large constant by increasing the running
time of the algorithm by a constant factor.

\subsection{A Simple Algorithm}

For ease of exposition, we start with a simple algorithm that is
relatively faithful to the proof of \thmref{upper-bound}.  
\begin{thm}\thmlabel{algorithm}
  For every constant dimension, $d$, and every set, $V$, of $n<\infty$
  points in $\R^d$, there exists a randomized $O(n\log^{d} n )$ time
  algorithm that constructs a graph $G'=(V,E)$ that has $O(n)$ edges and,
  with high probability, 
  \[
      \asf(G')=\begin{cases}
         1+ O((\log n/n)^{1/5}) & \text{for $d=2$} \\
         1+ O(1/n^{1/(4d-4)}) & \text{for $d\ge 3$} \\
      \end{cases}
  \]
\end{thm}

\begin{proof}
  Throughout this proof, a \emph{box} is an axis-aligned box of the form
  $[a_1,b_1]\times\cdots\times[a_d,b_d]$ and we call this a \emph{square
  box} if $b_i-a_i=b_j-a_j$ for all $i,j\in\{1,\ldots,d\}$. As before,
  $L(B)$ denotes the length of the longest side of the box, $B$.  The
  graph $G'$ constructed by our algorithm is similar to the graph, $G$,
  described in the previous section, except for three major differences:
  \begin{enumerate}
     \item Every element in $G$ that is defined in terms of a ball is
     now defined in terms of a square box.  In particular,
     \begin{enumerate}
       \item $\Delta_1,\ldots,\Delta_{n'}$ are minimal square boxes that contain
         the corresponding boxes defined by nodes of the fair-split tree;
       \item $D_1,\ldots,D_{n'}$ are square boxes where each $D_i$
         is a square box of side-length $cL(\Delta_i)$ centered at the
         center of $\Delta_i$; and
       \item $E_1,\ldots,E_{n'}$ are square boxes where each $E_i$
         is a square box of side-length $L(\Delta_i)/c$ that intersects $D_i$.
     \end{enumerate}
     \item The box $E_i$ does not perfectly maximize the number of points
       it contains. Instead, we guarantee that, if there exists any
       box of side-length $L(\Delta_i)/2c$ that intersects $D_i$ and
       contains $m\ge \epsilon n$ points of $V$, then $E_i$ is chosen so that
       it contains at least $m$ points of $V$.
     \item We change the value of $\epsilon$ to $\epsilon = 1/k$
      and use the values of $c$ and $k$ given in \remref{epsilon-choice}.
  \end{enumerate} 
  It is straightforward, but tedious, to check that the proof of
  \thmref{upper-bound} using these new definitions of $\Delta_i$,
  $D_i$, and $E_i$ shows that the graph $G'$ satisfies the
  requirements of the theorem.  What remains is to show how the graph $G'$
  can be constructed in $O(n\log^d n)$ time.

  As described in the proof of \lemref{k-partition}, computing the
  fair-split tree, and the resulting $\Delta_1,\ldots,\Delta_{n'}$,
  $V_1,\ldots,V_{n'}$, $u_1,\ldots,u_{n'}$ and $D_1,\ldots,D_{n'}$
  is easily accomplished in $O(n\log n)$ time.  The 2-spanner
  of the points in $V$ can be constructed in $O(n\log
  n)$ time using any of several different possible methods
  \cite{callahan.kosaraju:faster,salowe:constructing,vaidya:sparse}.
  An algorithm of Ruppert and Seidel \cite{ruppert.seidel:approximating}
  can, for any $\gamma >0$, construct a $(1+\gamma)$-spanner of
  $n'$ points in $\R^d$ that has $O((1/\gamma)^{d-1}n')$ edges in
  $O((1/\gamma)^{d-1} n'\log^{d-1} n')$ time.  Using this algorithm with
  $\gamma=1/k^{1/(d-1)}$, the $(1+1/k^{1/(d-1)})$-spanner of the $n'$
  points in $H$ can be constructed in time
  \[
     O(kn'\log^{d-1} n') = O(n\log^{d-1} n) \enspace .
  \]
  All that remains is to show how to compute $E_1,\ldots,E_{n'}$ and
  $w_1,\ldots,w_{n'}$ efficiently.

  The reason for moving from balls to square boxes is that boxes allow for
  the use of range trees \cite{bentley:multidimensional,luecker:data}.
  Range trees allow us to preprocess, using $O(n\log^{d-1} n)$ time and
  space, a set of $n$ points in $\R^d$ so that, for any query box, $B$,
  we can find, in $O(\log^{d-1} n)$ time: (1)~the number of points in $B$,
  or (2)~the point with minimum index in $B$.\footnote{For this third
  type of query, we use a constant-time range-minimum data structure
  \cite{bender.farach-colton:lca} as the 1-dimensional substructure.}
  Furthermore, by using a box-point duality that maps square boxes in
  $\R^d$ onto points in $\R^{d+1}$, a set of square boxes in $\R^d$
  can be preprocessed, using $O(n\log^d n)$ time and space, so that, for
  any query point $p\in\R^d$, we can, in $O(\log^d n)$ time, determine
  the index of the smallest box that contains $p$.

  Our first task is to find the boxes $E_1,\ldots,E_{n'}$ and for
  this we use random sampling.  Let $D'_{i}$ be the square box
  that is centered at the center of $\Delta_i$ and has side-length
  $(c+1/2c)L(\Delta_i)$.  We repeatedly select a random point $u\in V$.
  If $u$ is not in $D'_i$, then we discard it.  Otherwise, we count the
  number of points in the box, $E(u)$, that is centered at $u$ and has
  side length $L(\Delta_i)/c$.  The sample box containing the maximum
  number of points is chosen as $E_i$.  (If all samples were discarded,
  then we take $E_i$ to be any square box of side length $L(\Delta_i)/c$
  that intersects $D_i$.)

  Let $E'_i$, be some square box of side length $L(\Delta_i)/2c$ that
  intersects $D_i$ and that contains the maximum number of points in $V$.
  Suppose, furthermore, that $|E'_i\cap V|\ge \epsilon n$ (since this
  is the only case in which we make any guarantees about $E_i$).  If we
  repeat the above sampling procedure $\alpha\ln n/\epsilon$ times, then
  the probability that none of our samples is a point in $|E'_i\cap V|$
  is at most
  \[
      (1-\epsilon)^{\alpha\ln n/\epsilon} 
         \le 1/e^{\alpha\ln n} = 1/n^{\alpha} \enspace .
  \]
  Therefore, with high probability, at least one of our sample points,
  $u$, is in $E'_i$.  In this case, $E(u)$ contains $E'_i$, so we obtain
  a box, $E_i$ that intersects $D_i$ and contains $E'_i$.

  By storing the points of $V$ in a range tree, we can therefore identify
  the boxes $E_1,\ldots,E_{n'}$ in time
  \[
     O(n'(\log n/\epsilon)\log^{d-1} n) 
        = O\left(\frac{n\log^{d}n}{\epsilon k} \right) = O(n\log^d n) \enspace .
  \]

  Let $E\subseteq\{E_1,\ldots,E_{n'}\}$ be the subset of boxes that
  contain at least $\epsilon n$ points of $V$. By using the point-box
  duality, we can store $E$ in a range tree so that, for each point
  $w\in V$, we can determine the index, $i(w)$, of the smallest box
  that contains at least $\epsilon n$ points of $V$ and contains $w$.
  Building the range tree for (the duals of) $E_1,\ldots,E_{n'}$ takes
  $O(n'\log^d n)\subset O(n\log n)$ time and searching this range tree for
  each point $w\in V$ takes time
  \[
     O(n\log^d n) \enspace .
  \]

  Finally, we can store the point/index pairs $(u,i(u))$, for all
  $u\in V$, in a range tree so that, for each box, $E_i$, we can find
  a point $w_i\in E_i$ that minimizes $i(u)$.  Building this range
  tree takes $O(n\log^{d-1} n)$ time and the queries on this tree take
  $O(n'\log^{d-1} n)\subset O(n\log^{d-1} n)$ time.
\end{proof}

\subsection{A Faster Algorithm}

In the preceding section we gave an $O(n\log^d n)$ time algorithm.
In this section, we show that the running time can be reduced to $O(n\log
n)$ with only a small increase in the average stretch factor.

\begin{thm}\thmlabel{algorithm-faster}
  For every constant dimension, $d$, and every set, $V$, of $n<\infty$
  points in $\R^d$, there exists a randomized $O(n\log n)$ time
  algorithm that constructs a graph $G''=(V,E)$ that has $O(n)$ edges
  and, with high probability, 
  \[
      \asf(G'')=\begin{cases}
         1+ O((\log^{6} n/n)^{1/5}) & \text{for $d=2$} \\
         1+ O(\log^{(d-2)/(d-1)}n/n^{1/(4d-4)}) 
                  & \text{for $d\ge 3$} 
      \end{cases}
  \]
\end{thm}

\begin{proof}
  The construction of $G''$ is similar to the construction of
  $G'$ described in the proof of \thmref{algorithm}.  In the
  construction of $G'$, there are three issues that lead to a
  running-time of $\omega(n\log n)$: (1)~the construction of the
  $(1+(1/k)^{1/(d-1)})$-spanner of $H$ takes $\Theta(n\log^{d-1} n)$ time;
  (2)~the sampling algorithm used to find the boxes $E_1,\ldots,E_{n'}$
  that contain at least $\epsilon n$ points takes $\Theta(n\log^{d} n)$
  time; and (3)~determining the index, $i(u)$, of each point $u\in V$
  takes $\Theta(n\log^d n)$ time.  
  We address each of these issues in turn:  
  \begin{enumerate}
     \item We only construct a $(1+(\log^{d-2}
      n/k)^{1/(d-1)})$-spanner of $H$.  Using the algorithm of Ruppert
      and Seidel, the construction of this spanner takes only $O(n\log
      n)$ time.  This modification increases the average stretch factor
      of the resulting graph, so that the lower-order term increases by
      a factor of $\log^{(d-2)/(d-1)} n$.

     \item The sampling process used to find $E_1,\ldots,E_{n'}$
     has two phases.  In Phase~1, a range tree is constructed
     that contains the points of $V$.  In Phase~2, $O(n'\log
     n/\epsilon)$ queries are performed on this range tree.

     Phase~1 takes $O(n\log^{d-1} n)$ time.  To speed up Phase~1,
     we instead construct a range tree, $T_1$, on a Bernoulli sample
     $V'\subseteq V$ where each point is sampled independently with
     probability $p=\alpha/\log^{d-2} n$.  A standard application of
     Chernoff's Bounds \cite[Appendix~A.1]{alon.spencer:probabilistic}
     shows that, with high probability,
     \begin{enumerate}
       \item $T_1$ can be constructed in $O(n\log n)$ time;
       \item for all boxes, $B$, with $|B\cap V|\ge \epsilon n$,
       \[  
           (1/2)|B\cap V|\le  |B\cap V'|/p \le 2|B\cap V| \enspace;
       \]
       \item for all boxes, $B$, with $|B\cap V|\le \epsilon n/2$,
       \[    
           |B\cap V'|/p \le \epsilon n \enspace .
       \]
     \end{enumerate}
     Properties~(b) and (c) above ensure that the quantity $|B\cap V'|/p$,
     which can be computed in $O(\log^{d-1} n)$ time using $T_1$, is an
     accurate enough estimate of $|B\cap V|$.

     Phase~2 requires sampling $\alpha n'\ln n/\epsilon$ points and,
     for each sample point, $u$, counting the number of points of $V$
     in some box centered at $u$.  By increasing the value of $\epsilon$
     from $\epsilon=1/k$ to $\epsilon=\log^{d-1} n/k$, this counting can
     be done in $O(n\log n)$ time using $T_1$.  \remref{epsilon-choice},
     with the value $\kappa=d-1$, explains why this new choice of
     $\epsilon$ does not increase the average stretch factor for $d\ge
     3$, and increases it by a factor of $\log n$ for $d=2$.

    \item To determine the index, $i(u)$, of each point $u\in V$, we
      first construct a range tree, $T_2$ for (the duals of) some boxes
      in  $\{E_1,\ldots,E_{n'}\}$.  In $T_2$, we include every box,
      $B\in\{E_1,\ldots,E_{n'}\}$ such that $|B\cap V'|/p\ge\epsilon n$.
      From the preceding discussion, with high probability, $T_2$ then
      contains only boxes, $B$ such that $|B\cap V|\ge \epsilon n/2$ and
      $T_2$ includes every box, $B$ such that $|B\cap V|\ge 2\epsilon n$.

      Determining which boxes to include in $T_2$ requires $n'$ queries
      in $T_1$, so this takes $O(n'\log^{d-1} n)\subset O(n\log n)$ time.
      Building the tree $T_2$ takes $O(n'\log^{d} n)\subset O(n\log n)$
      time.  At this point, we would like to use $T_2$ to compute $i(w)$
      for every point $w\in V$, but this would take $\Omega(n\log^d
      n)$ time.  Instead, we take another sample, $V''\subseteq V$
      of size $n/\log^{d-1} n$ and compute $i(w)$ for each $w\in V''$.
      This takes only $O(n\log n)$ time.

      Finally, we put each pair $(u,i(u))$ for each $u\in V''$ into
      another range tree $T_3$.  This takes only $O(|V''|\log^{d-1}
      n)=O(n)$ time.  We then query $T_3$ with each box in
      $\{E_1,\ldots,E_{n'}\}$ to determine the point $w_i\in E_i\cap V''$
      that minimizes $i(w_i)$.  This takes only $O(n'\log^{d-1}n)\subset
      O(n\log n)$ time.

      It remains to argue that the point, $w_i$, which minimizes $i(w_i)$
      over all $w_i\in E_{i}\cap V''$ is a good-enough replacement for
      the point, $w_i^*$, that minimizes $i(w_i^*)$ over all $w_i^*\in
      E_i\cap V$.  To establish this, it is necessary to revisit the
      proof of \thmref{upper-bound}.  The only place in the proof of
      \thmref{upper-bound} in which the choice of $w_i$ plays a role is
      in counting the number of bad Type~4 pairs.

      Recall that this part of the proof of \thmref{upper-bound} assumes
      that the number of bad Type~4 pairs is greater than $\epsilon
      n^2$ and uses this assumption to derive a contradiction.
      In particular, the proof shows the existence of a point $w^*$
      that satisfies Conditions~B1--B3.  Walking through the proof, we
      see that it continues to work if there is any point $w^*\in V''$
      that satisifies Conditions~B1--B3.  In particular, $w^*\in V''$
      is sufficient to establish \clmref{zuper}, which is the only place
      in the entire proof of \thmref{upper-bound} that makes use of the
      fact that $w_i$ minimizes $i(w_i)$ over all $w_i\in E_i$.

      In the proof of \thmref{upper-bound}, the existence of $w^*$
      is established by the pigeonhole principle; we have subsets
      $S_1,\ldots,S_r$ of $V$ whose total size is $\Omega(\epsilon
      n^2/k)$ and we conclude that some element $w^*\in V$ must occur in
      $\Omega(\epsilon n/k)$ subsets.  However, the number of subsets,
      $r$, is only $O(n/k)$.  We can therefore make a stronger conclusion:
      There are $\Omega(\epsilon n)$ elements $w^*\in V$ such that $w^*$
      appears in $\Omega(\epsilon n/k)$ subsets.

      Thus, there exists a set $W^*\subseteq V$ of size $\Omega(\epsilon
      n)$ such that any point $w^*\in W^*$ is good enough to derive the
      contradiction required to bound the number of bad Type~4 pairs.
      If even one point of $W^*$ appears in $V''$, then the bound on the
      number of bad Type~4 pairs holds.  The sample, $V''$, is taken after
      the boxes $E_1,\ldots,E_{n'}$ and therefore after $W^*$ has been
      defined, so $V''$ and $W^*$ are independent.  The size of $W^*$
      is $\Omega(\epsilon n)$, so the probability that a randomly chosen
      element of $V$ is in $W^*$ is $\Omega(\epsilon)$. The sample $V''$
      contains $n/\log^{d-1} n$ randomly chosen elements of $V$, so the
      probability that $V''$ and $W^*$ are disjoint is at most
      \[
         (1-\Omega(\epsilon))^{n/\log^{d-1} n} 
            \le \exp(-\Omega(\epsilon n/\log^{d-1} n)) 
            \le n^{\alpha} \enspace ,
      \]
      for any constant $\alpha >0$ and all sufficiently large values
      of $n$.  Therefore, with high probability $V''$ contains at least
      one element of $W^*$ and the number of bad Type~4 pairs is at most
      $\epsilon n^2$.
  \end{enumerate}
  This completes the proof of \thmref{algorithm-faster}.
\end{proof}

\section{Lower Bounds}
\seclabel{lower-bound}

Next, we prove a simple lower-bound which shows that our bound on the
average stretch factor is at least of the right flavour: Combined with
the upper-bound of \thmref{upper-bound}, the following result shows that
the optimal bound is $1+O(1/n^{\delta})$ for some $\delta\in(1/(2d+1),1/2]$.

\begin{thm}\thmlabel{lower-bound}
  For every positive integer, $n$, there exists a set, $V$, of $n$
  points in $\R^2$, such that every geometric graph, $G$, with vertex
  set $V$ and having $O(n)$ edges has $\asf(G)\ge 1 + \Omega(1/\sqrt{n})$.
\end{thm}

\begin{proof}
  For simplicity, we assume $n$ is even.  The point set, $V$, has
  its points evenly distributed on two opposite sides of a square.
  The point set $V=A\cup B$, where
  \[  
      A = \{(0,i): i\in\{1,\ldots,n/2\}\} 
  \]
  and
  \[  
      B = \{(n/2,i): i\in\{1,\ldots,n/2\}\} 
  \]
  is an example (see \figref{lower-bound}.a).

  \begin{figure}
    \begin{center}
      \begin{tabular}{c@{\hspace{2cm}}c}
      \includegraphics{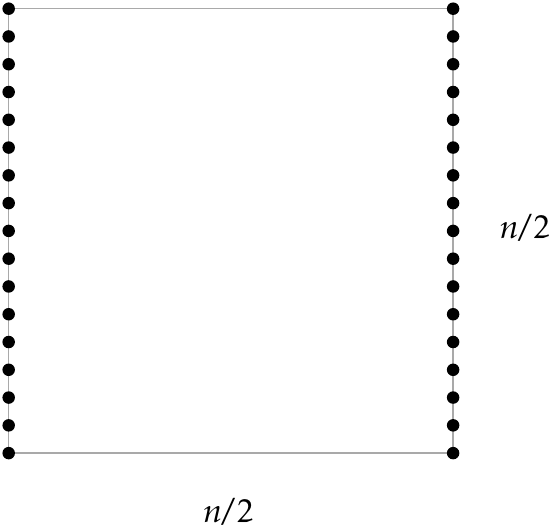} & \includegraphics{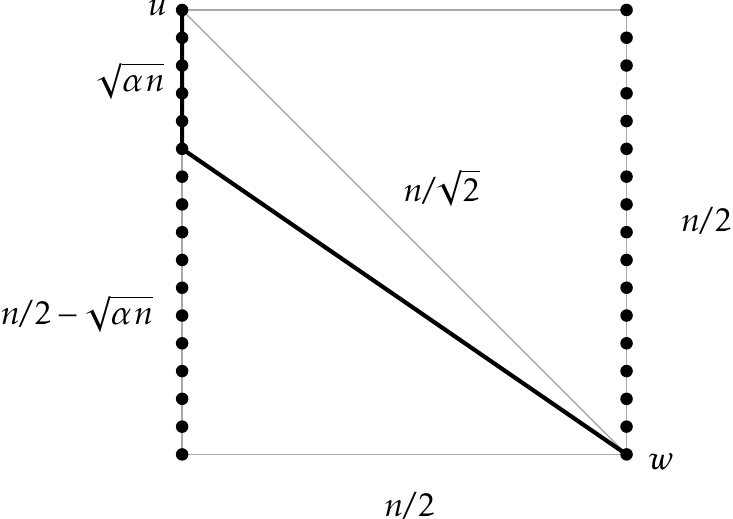} \\
      (a) &\hspace{1cm} (b) 
      \end{tabular}
    \end{center}
    \caption{The lower-bound (a)~point set for \thmref{lower-bound}, and
      (b)~the best-case ratio $\|uw\|_G/\|uw\|$ for a pair $(u,w)$ that
      is not covered by any edge.}
    \figlabel{lower-bound}
  \end{figure}

  Let $G$ be any graph with vertex set $V$.  We say that an edge $uw$
  with $u\in A$ and $w\in B$ \emph{covers} the set of pairs
  \[
     \{ \left(u+(0,i), w+(0,j)\right) : 
          i,j\in\{-\sqrt{\alpha n},\ldots,\sqrt{\alpha n}\}\}
  \]
  for some constant $\alpha$ to be discussed later.  Thus, any edge of
  $G$ covers at most $4\alpha n$ pairs in $A\times B$.

  Next, observe that if some pair of points $u\in A$ and $w\in B$ is
  not covered by any edge of $G$, then a straightforward minimization
  argument (see \figref{lower-bound}.b) shows that
  \begin{align}
     \frac{\|uw\|_G}{\|uw\|}
       & \ge \frac{\sqrt{\alpha n}+\sqrt{(n/2)^2+(n/2-\sqrt{\alpha n})^2}}
                {n/\sqrt{2}} \notag \\
       & = \frac{\sqrt{\alpha n}+\sqrt{n^2/2-\sqrt{\alpha}n^{3/2}+\alpha n/2}}
                {n/\sqrt{2}} \notag \\
       & \ge \frac{\sqrt{\alpha n}+n/\sqrt{2}-\sqrt{\alpha n/2}}
               {n/\sqrt{2}} \eqlabel{blobber} \\
       & \ge 1+\Omega(1/\sqrt{n}) \notag \enspace .
  \end{align}
  (Inequality~\eqref{blobber} is obtained by comparing 
   $\left(\sqrt{n^2/2-\sqrt{\alpha}n^{3/2}+\alpha n/2}\right)^2$ and 
   $\left(n/\sqrt{2}-\sqrt{\alpha n/2}\right)^2$.)
  If $G$ has $m\in O(n)$ edges, then we select $\alpha \le
  \binom{n}{2}/(8mn)$ so that
  \[  
     m4\alpha n \le \frac{\binom{n}{2}}{2} 
  \]
  In this way, at least half of the $\binom{n}{2}$ pairs of points in $V$
  are not covered by any edge and therefore,
  \[
     \asf(G) \ge 1 + \binom{n}{2}^{-1}\cdot\frac{\binom{n}{2}}{2}\cdot
          \Omega(1/\sqrt{n}) = 1 + \Omega(1/\sqrt{n}) \enspace . \qedhere
  \]
\end{proof}

We remark that the proof of \thmref{lower-bound} is easily modified to
provide a tradeoff between the number of edges of $G$ and the average
stretch factor.  In particular, if $G$ has $m\in o(n^2)$ edges, then
\[
   \asf(G) \ge 1 + \Omega(1/\sqrt{m}) \enspace .
\]

\section{Discussion}
\seclabel{discussion}

We have shown that, for any set, $V$, of $n$ points in $\R^d$, we can
construct, in $O(n\log n)$ time, a graph on $V$ whose average stretch
factor is $1+o_n(1)$.  Our construction consists of three parts, (1)~a
2-spanner, $G_2$ of $V$, (2)~a $1+o_n(1)$-spanner of a subset $H\subset V$
of so-called hubs, and (3)~a collection of edges that join all vertices
within each group, $V_i$, to a single representative vertex, $w_i$,
within a very dense (compared to $V_i$) set of vertices, $E_i\cap V$.

\subsection{Realistic Networks(?)}

We note that, if one has some form of reasonable well-distributed
assumption about the points in $V$, like that used by Aldous and Kendall,
then it is fairly straightforward to show that the first two parts of
our construction are sufficient to obtain an average stretch factor
of $1+o_n(1)$. The third part of our construction is only required
to deal with instances in which there are exponential differences in
density of points of $V$.  It seems natural to ask whether this part
of the construction is necessary in any real-world network or whether
the vertices of real-world networks are always well-distributed.

The first two parts of our construction are quite natural and can be
recognized in many real-world networks.  The most common (though admittedly,
imperfect) example is road networks where individual buildings are
interconnected by roads, which very often form a partial grid (a
$\sqrt{2}$-spanner). The cities, towns, and villages, containing these
buildings are themselves interconnected by a relatively fast, and fairly
direct, system of highways.

The third part of our construction seems less natural.  In road networks,
this part of the construction would correspond to densely-populated
areas that have several direct routes to them from well-separated
locations.  Since this part of our construction is only required to
deal with pathological cases involving inter-point distances that vary
exponentially, one might think that it does not appear in real-world
networks.

A quick inspection of the U.S. Highway System shows that, even by 1926,
there were many direct highway connections to Chicago, Detroit, Kansas
City, Memphis, Newport, and other large cities (see \figref{us-map}).
These highways were expensive and would have been built unless they
added real value to the road network.  While this does not correspond
perfectly with the third part of our construction, it does suggest that
a simple two-level network of clusters, each having a single hub, does
not produce real-world networks of sufficient quality; \emph{some}
form of extra augmentation is necessary.

\begin{figure}
  \begin{tabular}{@{}c@{\hspace{10pt}}c@{\hspace{10pt}}c@{\hspace{10pt}}c@{\hspace{10pt}}c@{}}
    \includegraphics[width=\dimexpr.2\textwidth-8pt]{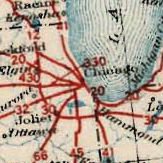} &
    \includegraphics[width=\dimexpr.2\textwidth-8pt]{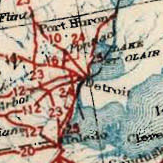} &
    \includegraphics[width=\dimexpr.2\textwidth-8pt]{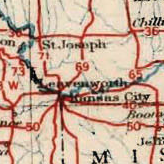} &
    \includegraphics[width=\dimexpr.2\textwidth-8pt]{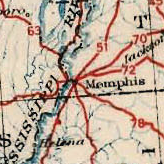} &
    \includegraphics[width=\dimexpr.2\textwidth-8pt]{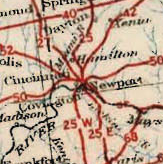}  \\
  \end{tabular}
  \caption{Large cities have many direct routes to them. (Source:
  Map of the final U.S. Highway system as approved November 11, 1926.
  This map is in the public domain.)}
  \figlabel{us-map}
\end{figure}

\noindent

\subsection{Open Problems}

Our results leave many areas open for further research.  We say that a
graph, $G=(V,E)$ has \emph{good average stretch} if $\asf(G)=1+o_n(1)$
and $|E|=O(n)$.

The following open problems have to do with strengthenings of
\thmref{upper-bound} in which $G$ has additional properties.

\begin{op}
  Given a point set, $V$, does there always exist a good average stretch
  graph $G=(V,E)$ whose total edge length is close to that of the minimum
  spanning tree?
\end{op}

\begin{op}
  Given a point set, $V$, does there always exist a good average stretch
  graph $G=(V,E)$ whose maximum degree is bounded by a constant?
\end{op}

\begin{op}
  Given a point set, $V$, does there always exist a good average stretch
  graph $G=(V,E)$ that is $k$-fault tolerant? That is, for any set
  $F\subset V$, $|F|\le k$, $\asf(G\setminus F)=1+o_n(1)$.
\end{op}

\begin{op}
  Bose \etal\ \cite{bose.dujmovic.ea:robust} define $f(k)$-robust spanners
  in terms of the (worst-case) stretch factor and their definition
  extends naturally to average stretch factor.  Given a point set, $V$,
  does there always exist a good average stretch graph $G=(V,E)$ that
  is $f(k)$-robust, for some reasonable function $f(k)$?
\end{op}

The following question asks if the upper-bound can be proven in a more
general setting:

\begin{op}\oplabel{metric-space}
  What conditions on a metric space $(V,d)$ are necessary and sufficient
  so that there always exist a graph $G=(V,E)$, $|E|\in O(n)$ with
  $\asf(G)=1+o_n(1)$?  (Here shortest paths in $G$ are measured in terms
  of the cost of their edges in the metric space.)
\end{op}

It seems likely that some of the techniques used to prove
\thmref{upper-bound} are applicable to metric spaces of bounded doubling
dimension \cite[Section~10.13]{heinonen:lectures}.  Is bounded doubling
dimension the weakest possible restriction on the metric space?  It is
clear that some restrictions on the metric space are required: In the
metric space in which all points have unit distance, the average stretch
factor of a graph, $G=(V,E)$, having $m$ edges is at least
\[
    \binom{n}{2}^{-1}\left(m + 2\left(\binom{n}{2}-m\right)\right) 
\]
since, if there is no edge between $u$ and $w$ in $G$, then $\|uw\|_G
\ge 2$.  Therefore any graph with average stretch factor $1+o_n(1)$
must have $\binom{n}{2}-o(n^2)$ edges.

\section*{Acknowledgement}

The authors of this paper are partly funded by NSERC and CFI.  The author
are indebted to Shay Solomon for providing helpful feedback on an earlier
version of this paper.

\bibliographystyle{abbrvurl}
\bibliography{avgstretch}

\section*{Authors}

\noindent
\includegraphics[width=.31\textwidth]{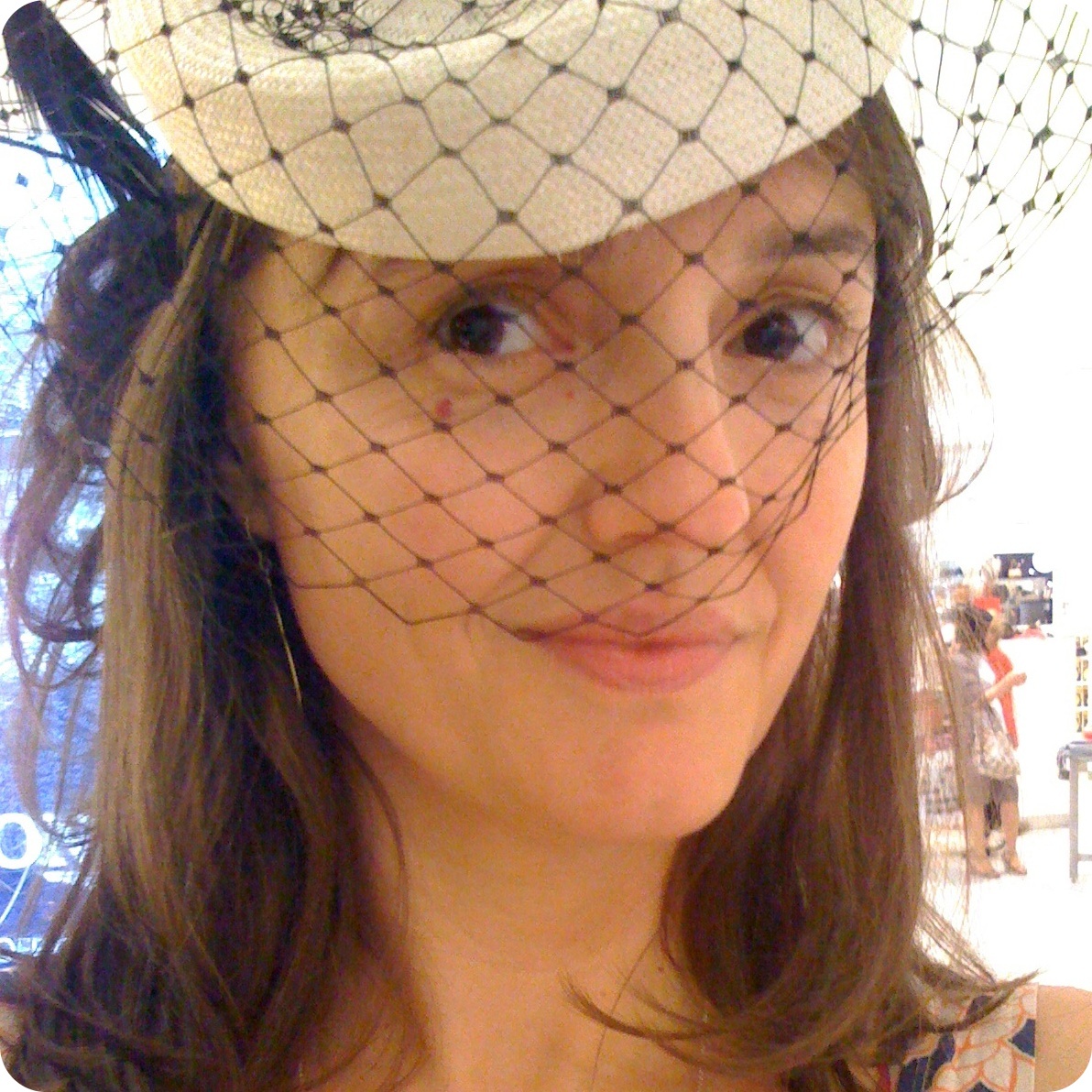}%
\hspace{.035\textwidth}%
\includegraphics[width=.31\textwidth]{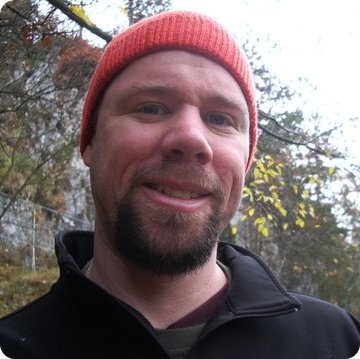}%
\hspace{.035\textwidth}%
\includegraphics[width=.31\textwidth]{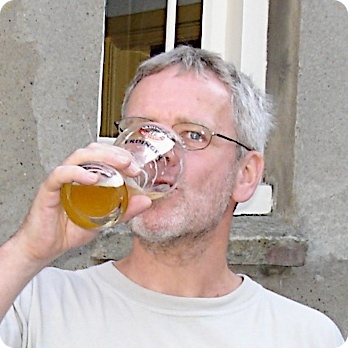}%

\noindent\emph{Vida Dujmovi\'c.}
School of Mathematics and Statistics and Department of Systems and Computer Engineering, Carleton University

\noindent\emph{Pat Morin} and \emph{Michiel Smid.}
School of Computer Scence, Carleton University

\end{document}